\documentclass[pdftex,journal]{IEEEtran}

\usepackage[dvips]{graphicx}
\usepackage{graphics}
\usepackage{amsmath,amsfonts,amssymb,amsthm}
\usepackage{cite}
\usepackage{url}
\usepackage{epstopdf}

\newtheorem{theorem}{Theorem}

\title{Upstream Polling Protocols for Flow Control in
PON/xDSL Hybrid Access Networks}

\author{Anu Mercian, Elliot I.~Gurrola, Frank Aurzada, Michael P. McGarry,
 and Martin Reisslein
\thanks{A.~Merican and M.~Reisslein are with the School of Electrical,
Computer,
and  Energy Eng., Arizona State Univ., Tempe, AZ 85287-5706, USA
(e-mail: \{amercian, reisslein\}@asu.edu).}
\thanks{E.~Gurrola and M.~McGarry are with the Dept. of Electr. and Comp.
Eng., University of Texas at El Paso, El Paso, TX 79968, USA
(email: mpmcgarry@utep.edu).}
\thanks{F.~Aurzada is with the Dept. of Mathematics,
Techn. Univ. Darmstadt, 64289 Darmstadt, Germany,
(e-mail: aurzada@mathematik.tu-darmstadt.de).}}

\begin{document}

\maketitle

\begin{abstract}
In a hybrid PON/xDSL access network, multiple Customer Premise Equipment (CPE)
nodes connect over individual Digital Subscriber Lines (DSLs) to a
drop-point device. The drop-point device, which is typically
reverse powered from the customer, is co-located with
an Optical Network Unit (ONU) of the Passive Optical Network (PON).
We demonstrate that the drop-point experiences
very high buffer occupancies when no flow control or standard
Ethernet PAUSE frame flow control is employed.
In order to reduce the buffer occupancies in the drop-point,
we introduce two gated flow control protocols that extend the
polling-based PON medium access control to the DSL segments between
the CPEs and the ONUs.
We analyze the timing of the gated flow control mechanisms to
specify the latest possible time instant when CPEs can start the DSL
upstream transmissions so that the ONU can forward the upstream
transmissions at the full PON upstream transmission bit rate.
Through extensive simulations for a wide range of bursty traffic models,
we find that the gated flow control mechanisms, specifically, the
ONU and CPE grant sizing policies, enable effective control of the maximum
drop-point buffer occupancies.
\end{abstract}

\begin{keywords}
Buffer occupancy, Digital Subscriber Line (DSL), Flow control,
Medium access control, Optical Network Unit (ONU),
Passive Optical Network (PON), Polling protocol.
\end{keywords}

\section{Introduction}
\label{sec:intro}
Access networks are communication networks that interconnect private local
area networks, such as the networks in the homes of individuals, with public
metropolitan and core networks, such as those constructed by service providers
to connect paying subscribers to the Internet.
Private local area networks often employ high speed wired and wireless
communications technologies, such as IEEE 802.3 Ethernet (up to 1 Gbit/sec)
and IEEE 802.11 WiFi (up to 600 Mbit/sec). These high-speed communications
technologies are cost effective in private local area networks due to
the short distances involved and subsequent low installation costs.
Public metropolitan and core networks employ a variety of communication
technologies that include dense wavelength division multiplexed technologies
over fiber optic transmission channels (up to 1 Tbit/sec). These high-speed
communication technologies are cost effective due to the
cost sharing over many paying subscribers. Access networks require
significantly higher installation costs compared to private local area
networks due to larger distances that must be covered. At the same time,
access networks have significantly smaller degrees of cost sharing compared
to public metropolitan and core networks; thereby increasing cost per paying
subscriber. As a result, access network technologies must keep installation
costs low~\cite{K0506}. Utilizing existing bandwidth-limited copper wire or
shared optical fiber will keep installation costs low~\cite{maz2014coe}.

In this paper we present our study of hybrid access networks that utilize
both copper wire and shared optical fiber. The shared optical fiber extends
from the service provider's central office to a drop-point whereby the final
few hundred meters to the subscriber premise are reached by existing
twisted-pair
copper wire. Figure \ref{fig:pondslnet} illustrates this hybrid access network
architecture that leverages the installation cost benefits of existing copper
wire and the latest advances in digital subscriber line (DSL)~\cite{dslbook}
technology that can realize up to 1 Gbps over short distances of twisted-pair
copper wire.
\begin{figure*}[t]
\centering
\includegraphics[scale=0.6]{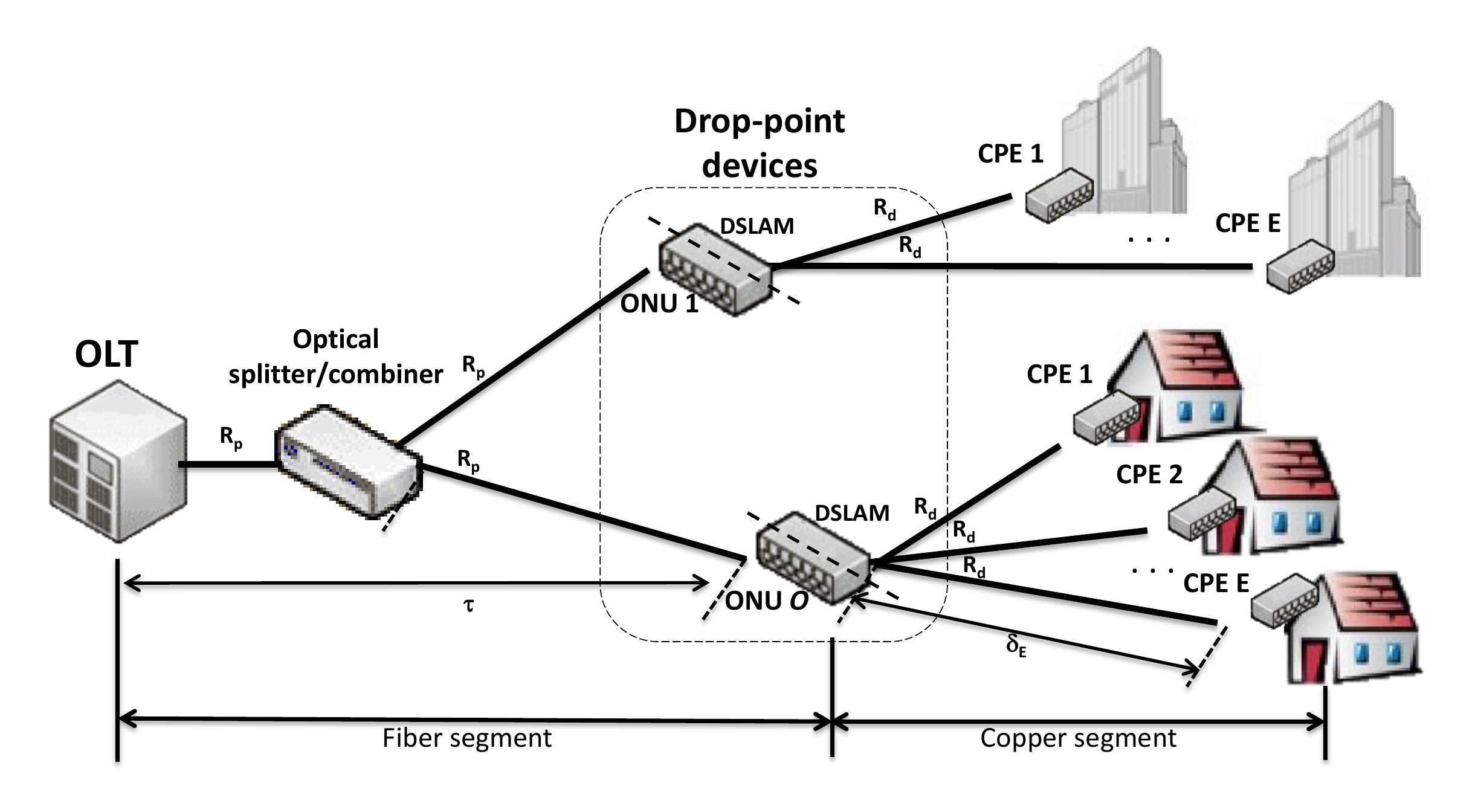}
\caption{A hybrid PON/xDSL access network architecture consists of a
  passive optical network (PON) connected to multiple copper digital
  subscriber lines (DSLs). The PON OLT
  connects to several drop-point devices. Each drop-point device is a
  combined PON Optical Network Unit (ONU) and DSL Access Multiplexer
  (DSLAM). Through the DSLAM, each drop-point device connects to
  multiple subscriber DSL customer premise equipment (CPE) nodes.}
\label{fig:pondslnet}
\end{figure*}

The optical fiber segment of this hybrid access network is organized as a
shared passive optical network (PON), whereby multiple optical network
units (ONUs)
share a single optical fiber connected to an optical line terminal (OLT) at the
service provider central office. The copper segments begin at each ONU whereby
the fiber is dropped and existing copper wires are utilized via DSL
transmission
technology to reach each subscriber premise. Each ONU is coupled with a DSL
access multiplexer (DSLAM) at the fiber drop-point. This so-called drop-point
device is active and therefore requires electric power to operate. However,
service providers want these devices to maintain the \textit{deploy-anywhere}
property of the optical splitter/combiner in a typical PON. To maintain this
property, each drop-point device is reverse powered using a subscriber's power
source. For this reason, it is of critical importance to reduce the energy
consumption of this device.

Reducing the memory capacity of the drop-point is an option for reducing
its energy consumption. A drop-point with a small memory capacity
translates into a design with a smaller memory device that contains
fewer transistors and capacitors that consume energy. However, reducing
the memory capacity of a drop-point can result in significant packet loss
if measures are not taken to back-pressure the buffering into either the
OLT in the downstream direction, or the DSL customer premise equipment (CPE) in
the upstream direction. The magnitude of buffering that can occur at the
drop-point is quite large due to the transmission bit rate mismatch between
the DSL line and the PON. Flow control mechanisms are, therefore, required to
avoid significant packet loss. In this paper we specifically examine several
upstream polling strategies for controlling the flow of upstream data from
each CPE to its associated drop-point device~\cite{L0313}. The objective of
these strategies is to minimize the maximum buffer occupancy required at
each drop-point with very low or no packet loss.

\subsection{Background}
Providing digital data communication through the access network
emerged with Digital Subscriber Loop or Line (DSL) technology in the
late 1970s and early 1980s~\cite{ABG1181}. At that time, researchers
identified mechanisms to aggregate digital data signals with analog
telephony signals and identified effective power levels and coding
mechanisms to tolerate the transmission impairments of the copper
loops used for analog telephony. These impairments included signal
reflections, cross-talk, and impulse noise~\cite{ABG1181}.  Recent
efforts exploit multiple-input-multiple-output (MIMO) or vectoring
techniques to cancel the crosstalk
impairment~\cite{GC0602,LCJM0907}. Systems using these techniques can
achieve approximately 1 Gbps transmission using four twisted pairs
across distances up to 300 meters~\cite{LCJM0907}. The recently
developed G.fast~\cite{Gfast,TGNM0813} DSL standard utilizes vectoring
techniques to achieve up to 1 Gbps speeds over these short distances.

Passive optical networks were envisioned in the late 1980s and early
1990s as an alternative to copper transmission between service
provider central offices and subscriber
premises~\cite{SBFHPO1187,FIMDDDKKDPB1194}. A PON utilizes a shared
fiber optic transmission medium shared by up to a few dozen
subscribers thereby reducing per-subscriber installation
costs. Further, PONs employ passive devices between the service
provider's central office and the subscriber premises to also reduce
recurring operational costs. Standardization of PON technologies began
around the early 2000s (e.g., Ethernet PONs~\cite{KP0202}) and have
subsequently achieved widespread deployment in the past few
years. Each of the various PON standards has considered the dynamic
bandwidth allocation (DBA) algorithms that decide how various
subscribers share the bandwidth of the optical fiber out of scope. As
a result, research activity on DBA algorithms started around the time
the standards were being
developed~\cite{dias2015off,gra2014max,KMP0202,sar2014ifa,SCAWM0309,tur2015new,ZM0709}.

Hybrid access network designs combine several transmission media types
(e.g., fiber, copper,
free space)~\cite{ahm2015ser,fan2015dem,KDRC0512,L0313} to reach
subscribers.
Hybrid fiber and copper access networks~\cite{L0313,gau2014unb}
provide a good balance between the increased bandwidth of fiber optic
transmission and the cost benefits of using already deployed copper
transmission lines. Wireless technologies in access networks
add both a very low-cost installation option by using free space transmission
as well as mobility features for users.

\subsection{Related Work}
Although there is significant literature on the integration of PONs
with wireless transmission media, e.g., WOBAN~\cite{SDM1107} and
FiWi~\cite{ahm2012rpr,GMA0209}, there is a dearth of literature on the
integration of PONs with copper transmission media.

Around the time the various PON standards were being developed,
researchers proposed developing hybrid PON/xDSL access networks. These
hybrid access networks would utilize DSL transmission technologies
with existing twisted-pair copper wire in conjunction with PONs. In
\cite{CB0300,CBABB0601} an early PON standard called ITU-T 983.1
Broadband PON (BPON) was coupled with VDSL to reach subscribers in a
cost-effective manner. Specifically,
an architecture for a combined ONU/VDSL line card (drop-point) device
that bridged a single VDSL line onto the PON was described in~\cite{CB0300}.
A full demonstration system for transferring MPEG-2 video through a
BPON/VDSL network using the ONU/VDSL line card~\cite{CB0300} was
presented in~\cite{CBABB0601}. In \cite{R0903}, a mathematical model
of the number of VDSL subscribers that can be serviced by a single
ONU as a function of a few VDSL parameters (e.g., symmetric
operation and bit rates) was presented. This model can help service
providers design their PON/xDSL networks to support the desired
number of subscribers. In a study on QoS-aware intra-ONU scheduling
for PONs~\cite{AYD1203}, hybrid PON/xDSL access networks were noted
as a promising candidate for cost-effective broadband access. This
early work on hybrid PON/xDSL access networks demonstrated its
feasibility and provided some analysis for capacity planning but
ignored detailed design elements of the drop-point device that
bridges the PON with the various DSL lines connecting to
subscribers.

Two physical-layer systems to bridge VDSL signals over a fiber
access network were proposed in~\cite{LTTBAPW0805}. Individual VDSL
signals are converted to be spectrally stacked into a composite
signal that modulates an optical carrier.  In the first system the
optical carrier is supplied by a laser at the ONU and in the second
system the optical carrier is supplied by a laser in the OLT that is
reflected and modulated by a Reflective Semiconductor Optical
Amplifier (RSOA) at the ONU. The optical carrier provides 1 GHz of
spectral width accommodating 40 VDSL lines without guard bands and
25 VDSL lines with guard bands. Although, this approach to a hybrid
PON/xDSL allows the drop-point device to avoid buffering as well as
contain simple logic by pulling the DSLAM functionality into the
OLT, the design requires the PON to carry the full bandwidth of each
VDSL line even when idle. Designs that operate at the link layer
rather than physical layer can avoid transmission of idle data on
the PON thereby increasing the number of subscribers that can be
supported by capitalizing on statistical multiplexing gains.

The coaxial copper cable deployed by cable companies represents another existing
copper technology that can be used in conjunction with PONs to create
a hybrid access network.
Such a hybrid access network combining an Ethernet PON with an Ethernet over
Coax (EoC) network was proposed in~\cite{WLZZC0909}.
The proposed network uses EPON protocols on the EoC
segment in isolation from the EPON segment without any coordination between
the segments. A similar network was examined in~\cite{W1009} in terms of
the blocking probability and delay for a video-on-demand
service. None of these studies discussed the design of the drop-point device
or explored DBA algorithms for these types of networks.

In November 2011, the IEEE 802.3 working group initiated the creation of
a study to extend the EPON protocol over hybrid fiber-coax cable television
networks; the developing standard is referred to as EPON Protocol over Coax
(EPoC)~\cite{BTMZCEF1013}. Developing bandwidth allocation schemes for
EPoC has received  little research attention to date.
In particular, a DBA algorithm that increases
channel utilization in spite of increased propagation delays due to
the coaxial copper network was designed in~\cite{BTZCEFM0413}.
Mechanisms to map Ethernet frame transmissions to/from
the time division multiplexed channel of the PON to the time and frequency
division multiplexed coaxial network have been studied
in~\cite{BTMZCEF1013,BTZCEFM1014}.

\subsection{Our Contribution}
In this paper we contribute the first hybrid PON/xDSL drop-point design
providing lowered
energy consumption by means of reduced buffering requirements. We
mitigate the packet loss effects of the small drop-point buffers by defining
and evaluating several polling strategies that contain flow control.
Although we focus on xDSL as the copper technology in the hybrid
access networks, our proposed flow control polling protocols can be
analogously employed with other copper technologies, such as coax cable.

We define polling mechanisms that place the DSL CPEs under the control
of the PON OLT. With this flow control mechanism the polling MAC protocols that
have been designed for PONs are extended to a second stage of polling in the DSL
segments. We call this mechanism GATED flow control as the OLT on the
PON not only grants transmission access to ONUs on the PON but
determines when DSL CPEs transmit upstream to their attached ONUs. As
far as we know, we are the first to explore joint upstream transmission
coordination for hybrid PON/xDSL access networks.

The work presented in this paper provides significant extensions to the work
we presented at two conferences~\cite{MGL1213,MGMR1115}. In~\cite{MGL1213},
we presented a preliminary form of one of the two Gated flow control
mechanisms along with some initial simulation results. In~\cite{MGMR1115},
we present simulation results for one DBA algorithm,
namely (Online, Limited)~\cite{KMP0202,ZM0709}.
In contrast, in this article we comprehensively specify two
Gated flow control protocols through detailed analysis of the CPE
transmission timing (scheduling) and present extensive simulation results
that include the (Online, Gated) and (Online, Excess) DBA algorithms.

\section{PON/xDSL Network}  \label{sec:polling}
In this section, we briefly describe the PON/xDSL network architecture
and outline flow control based on conventional PON polling in
conjunction with the standard Ethernet PAUSE frame.

\subsection{Network Architecture}
As illustrated in Figure \ref{fig:pondslnet}, a PON/xDSL hybrid access
network connects multiple CPE devices $c,\ c = 1, 2, \ldots, E$, each
via its own DSL, to a drop-point device. Let $R_d$ [bit/s] denote the
upstream transmission bit rate on each DSL line and $\delta_c$ denote the
one-way propagation delay [s] between CPE $c$ and its drop-point;
the main model notations are summarized in Table~\ref{tab:params}.

Each drop-point consists of a DSLAM combined with a ONU of the
PON. Let $O$ denote the total number of ONUs in the PON; whereby each ONU
is part of a drop-point, $R_p$ be the upstream transmission bit rate [bit/s]
from an ONU to the PON OLT, and $\tau$ be the one-way transmission
delay [s] between an ONU and the OLT. We note that typically $R_p > R_d$.
\begin{table}[t]
\caption{Model Parameters for PON/xDSL Hybrid Access Network}
\label{tab:params}
\centering
\begin{tabular}{|l|l|}
\hline
Param. & Meaning \\ [0.5ex]
\hline
\multicolumn{2}{|c|}{Network structure} \\
$R_{d}$ &   xDSL upstream transmission bit rate [bit/s] \\
$R_{p}$ &   PON upstream transmission bit rate [bit/s]\\
$E$ &  Number of CPEs per ONU; CPE index $c,\ c = 1, 2, \ldots, E$\\
$\delta_c$  &   One-way propagation delay from CPE $c$ to drop-point [s] \\
$\tau$   &   One-way propagation delay between OLT and ONU [s]\\ \hline
\multicolumn{2}{|c|}{Polling protocol} \\
$g_p$  &    Transmission time [s] for grant message on downstream PON\\
$g_d$   &   Transmission time [s] for grant message on downstream DSL \\
$G_c$   &   Size of upstream transmission window [bits] granted to
            CPE $c$ \\
$M$ &   Maximum packet size [bits] \\  \hline
\multicolumn{2}{|c|}{Polling analysis for individual CPE $c$} \\
$\sigma_c$  & Start time instant of CPE $c$ upstream DSL transmission \\
 & \ \ \    (relative to start time of a cycle) \\
$\alpha_c$ & Time instant when CPE $c$ data
            starts to arrive at drop-point \\
$\omega_c$ & Time instant when CPE $c$ data is compl.\ received
    at drop-point \\
$\mu_c$  & Time instant when ONU starts to transmit (serve) CPE $c$
       data \\
   & \ \ \ (= time instant of max. CPE $c$ drop-point buffer occupancy) \\
$\beta_c$ &  Time instant when ONU upstream transm. of CPE $c$ data
         ends \\
$T$ & Cycle duration from start instant of OLT grant transmission \\
   & \ \ \ to receipt of CPE data by OLT \\ \hline
\multicolumn{2}{|c|}{Segregated CPE transmissions on PON} \\
$\mu_{(E)}$ & Start time of ONU transm. of
   back-to-back \\
    & \ \ CPE $1, 2, \ldots, E$ data \\
$\sigma_c^s$ & Start time of CPE $c$ upstream DSL transmission \\ \hline
\multicolumn{2}{|c|}{Multiplexed CPE transmissions on PON} \\
$\mu^m$  & Start time of ONU transm. of multiplexed CPE data \\
$\sigma_c^m$ & Start time of CPE $c$ upstream DSL transmission \\ \hline
\end{tabular}
\end{table}

To support the ``deploy-anywhere'' property, each drop-point device is
remotely powered over the DSL using the power supply of several
subscribers. As a result of the remote powering, the drop-point design
must consume as little energy as possible. We explore reducing
buffering at the drop-point to reduce energy consumption. By reducing
the \textit{maximum buffer occupancy}, the drop-point can be designed
with a reduced memory capacity that will translate into fewer energy
consuming transistors and/or capacitors.
We utilize flow control
strategies through MAC polling to control buffering at each
drop-point. We introduce three upstream polling strategies that
provide flow control:
\begin{enumerate}
  \item ONU polling with PAUSE frame flow control
  \item Gated ONU:CPE polling flow control
    with segregated CPE transmission on PON (\textit{\textbf{ONU:CPE:seg}})
  \item Gated ONU:CPE polling flow control
    with multiplexed CPE transmission on PON (\textit{\textbf{ONU:CPE:mux}})
\end{enumerate}

\subsection{ONU Polling with PAUSE-Frame Flow Control}
\label{sec:onupoll}
Our first proposed upstream polling strategy utilizes OLT media access
control (MAC) through polling only on the PON segment. With this strategy,
each CPE continuously
transmits upstream on its attached DSL. To control the flow of
upstream traffic so as to reduce the maximum buffer occupancy,
we utilize the standard Ethernet PAUSE frame flow control:
When an Ethernet receiver's buffer reaches a certain threshold that
Ethernet node transmits a PAUSE frame to the attached node in a
point-to-point configuration. Upon receipt of the PAUSE frame, an
Ethernet transmitter squelches its transmission for the time period
indicated in the PAUSE frame.
In the PON/xDSL network, the drop-point monitors its upstream
DSL buffer and once its occupancy reaches a certain threshold,
the drop-point transmits
a PAUSE frame downstream to the DSL CPE. When the DSL CPE
receives the PAUSE frame it squelches its transmission for the time
period indicated in the PAUSE frame.

\section{Gated ONU:CPE Polling Flow Control}  \label{onucpe_poll:sec}
\subsection{Overview of ONU:CPE Polling Protocol}
Our proposed upstream ONU:CPE polling strategies extend the
OLT MAC polling~\cite{hos2015mul,kan2012ban,KMP0202,mer2011ban,ZM0709}
to each DSL CPE.
A DSL CPE transmits upstream only when explicitly
polled by the PON OLT with a GATE message.
The PON OLT conducts two stages of polling,
the first stage polls each ONU and the second stage polls each CPE.
More specifically, in a given cycle, the OLT sends a gate
message to the ONU to grant the ONU an upstream transmission window for
the data and bandwidth requests (reports) from the attached CPEs as well as
$E$ gate messages for the ONU to forward to the attached $E$ CPEs.
We denote $g_p$ for the downstream transmission time of a gate message on
the PON and $g_d$ for the downstream transmission time of a gate message
on a DSL. Moreover, we denote $G_c$ for the size [bit] of the upstream
transmission window granted to CPE $c$.
By controlling the transmission of each DSL CPE, the PON OLT can exercise
tight control over the magnitude of buffering that occurs at the drop-point.

\subsection{CPE Grant Sizing} \label{cpesize:sec}
In ONU:CPE polling, the OLT can apply any of the existing ONU grant
sizing strategies~\cite{KMP0202,ZM0709,MR0712} to assign each ONU an upstream
transmission window duration (grant size) according to the reported
bandwidth requests.
In turn, the OLT allocates a given ONU grant size to grants to the
attached CPEs and other (non-xDSL traffic) at the ONU.
When making a grant sizing decision for an ONU, the OLT knows the
bandwidth requests from all CPEs attached to the ONU.
Thus, the OLT can employ any of the grant sizing approaches requiring
knowledge of all bandwidth requests,
i.e., so-called offline approaches~\cite{ZM0709,MR0712},
for sizing the CPE grants.

As specified by the VDSL standard~\cite{VDSL}, Ethernet frames are
encapsulated in a continuous stream of Packet Transfer Mode (PTM) 65~Byte
codewords, see~\cite[Annex N]{ADSL}. Each codeword
contains one synchronization byte for every 64 bytes of data as well
as control characters and idle data bytes. The VDSL CPE under study
has been designed to suppress PTM codewords that contain all idle
data bytes. However, Ethernet frames can be encapsulated in PTM codewords
that contain idle data bytes. The number of bytes to be transmitted to
release a certain number of intended Ethernet frames from the CPE
depends on how the individual Ethernet frames expand within the PTM
codewords due to the inclusion of both PTM control characters and idle
data bytes. Modeling the exact number of bytes consumed by PTM codewords for
a given number of Ethernet frames requires knowledge of the individual
Ethernet frame sizes.  That information is not available at the
OLT. Therefore, we estimate the CPE grant size to accommodate the PON
grant size with one synchronization byte for every 64 data bytes. We
then assume one extra codeword to contain control characters and idle
data bytes.

Due to the CPE grant size estimation, it is possible
that the CPE grant is too small and therefore does not allow all of
the PTM codewords containing the intended Ethernet frames to be
transmitted. In this case, an intended Ethernet frame will only be
partially received at the ONU with the other part left at the CPE. With
the next grant, the remainder of this Ethernet frame will be
transmitted, along with the other Ethernet frames intended for that
grant. The resulting extra Ethernet frame at the ONU will not
be accommodated by the current PON upstream grant. That Ethernet frame becomes
residue that stays at the drop point until it can be serviced in the next PON
upstream grant.
We also note that if we increased our CPE grant size estimate,
then the grant would be
too large and result in one or more Ethernet frames left as residue at
the ONU because the PON upstream grant would not accommodate them.

In the subsequent analysis of ONU CPE polling in this
Section~\ref{onucpe_poll:sec}, we neglect the drop point buffer residue. The
simulations in Section~\ref{sec:experiment} consider the full xDSL
and PON framing details and thus include the effects of the residue.
We note that due to neglecting the residue,
the PON delay analysis in Section~\ref{pondelay:sec} is approximate.
However, we emphasize that the timing (scheduling) analyses in
Sections~\ref{cpesepspec:sec} and~\ref{onucpemux:sec} are accurate
for the grant sizes determined by the OLT.

\subsection{Basic Polling Timing Analysis for an Individual CPE}
\label{basicindcpe:sec}
In this subsection we examine the timing of the polling of a single
CPE $c$ attached to an ONU. We establish basic timing relationships
of the CPE and ONU upstream data transmissions.
Due to the transmission delays of the ONU and CPE grant messages
and the downstream propagation delays, the CPE can start transmitting
at the earliest at time instant
\begin{eqnarray} \label{sigmac:eqn}
\sigma_c = 2 g_p + \tau + g_d + \delta_c.
\end{eqnarray}
Note that we measure time instants relative to the beginning of the
cycle, i.e., we consider the time instant when the OLT begins to transmit
the gate message downstream as zero.
For the basic analysis we assume that the CPE begins to transmit its
data at this earliest possible time instant $\sigma_c$ to the
drop-point.

\begin{figure}[t]
\centering
\includegraphics[scale=0.35]{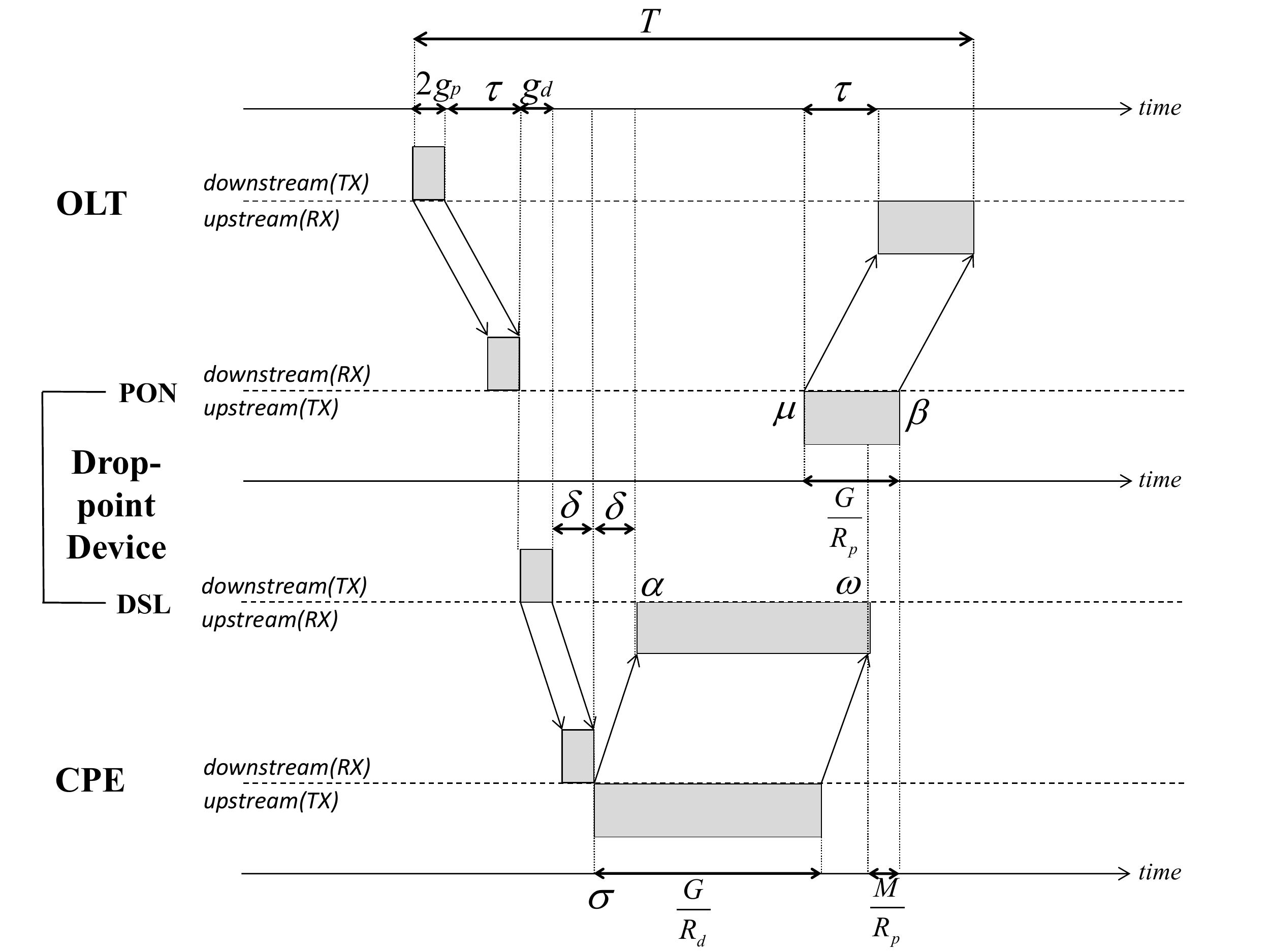}
\caption{Illustration of polling timing for an individual CPE $c$.}
\label{fig:onecpepolltime}
\end{figure}
As illustrated in Figure~\ref{fig:onecpepolltime}, a CPE upstream
transmission grant of size $G_c$ needs to be transmitted through
both the DSL segment (CPE $\to$ drop-point) and the PON segment
(drop-point $\to$ OLT). To determine when the transmission on the
PON should begin, we must consider that the last bit of a packet
must have arrived at the drop-point device from a CPE before the
first bit of that same packet can be transmitted by the ONU to the
OLT. We let $M$ denote the maximum packet size [in bit] and
conservatively consider maximum size packets in the following
analysis. Focusing on the last packet of the CPE upstream
transmission, we note that the end of the last packet, i.e., the end
of the CPE upstream transmission must be received by the drop-point
before the ONU can forward this last packet over the PON to the OLT.
We denote $\alpha_c$ for the time instant when the CPE upstream
transmission begins to arrive (and occupy buffer space) at the drop
point, i.e.,
\begin{eqnarray}
 \alpha_c = \sigma_c + \delta_c.
\end{eqnarray}

After complete receipt of the last packet at time instant
\begin{eqnarray}  \label{omega:eqn}
\omega_c  = \alpha_c + \frac{G_c}{R_d},
\end{eqnarray}
the ONU can immediately transmit this last packet
to the OLT. We denote $\beta_c$ for the time instant when
the last packet is completely transmitted by the ONU, i.e.,
when the CPE transmission stops to occupy buffer in the drop-point.
Clearly,
\begin{eqnarray}
\beta_c = \omega_c + \frac{M}{R_p}.
\end{eqnarray}
The end of the last packet reaches the OLT after the PON propagation delay,
resulting in the cycle duration $T = \beta_c + \tau$.

For the last packet to be able to start ONU transmission at time instant
$\beta_c - M/R_p$, all preceding packets must have already transmitted
by the ONU by time instant $\beta_c - M/R_p$.
More generally, the ONU finishes the transmission of the
$G_c$ bits of CPE data by instant $\beta_c$,
if the ONU starts the PON upstream transmission
(service) of the CPE data at time instant
\begin{eqnarray}  \label{muc:eqn}
\mu_c = \beta_c - \frac{G_c}{R_p}.
\end{eqnarray}

We note that throughout this study we consider polling strategies
that transmit CPE data at the full optical transmission bit rate
$R_p$ on the PON upstream channel from ONU to OLT. Since the xDSL
transmission bit rate $R_d$ is typically lower than the fiber
transmission bit rate $R_p$, the drop point needs to buffer a part
of a CPE data transmission, which is received at rate $R_d < R_p$ at
the drop point, before onward transmission at rate $R_p$ over the
PON. Polling strategies that transmit on the PON upstream channel at
a rate lower than $R_p$ can reduce drop point buffering at the
expense of increased delay. The study of such strategies that only
partially utilize the optical upstream transmission bit rate is left
for future research.

\begin{figure}[t]
\centering
\includegraphics[scale=0.5]{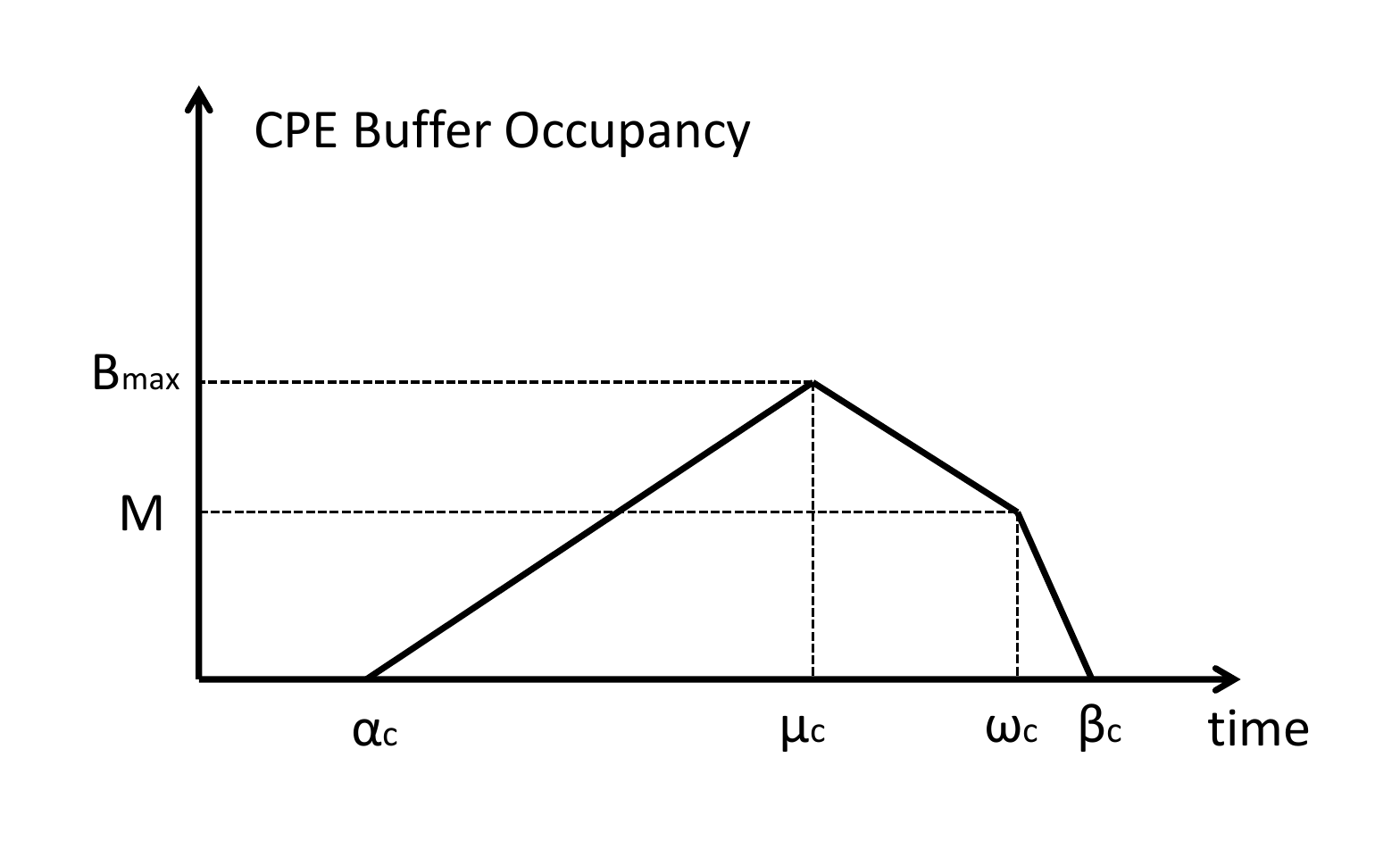}
\caption{Illustration of buffer occupancy for a given CPE $c$ in drop point:
The CPE buffer is filled at rate $R_d$ until the ONU starts transmitting
the CPE data at time instant $\mu_c$ with rate $R_p > R_d$. Then,
the buffer occupancy decreases at rate $R_p - R_d$ until the CPE data stops
arriving to the drop point at instant $\omega_c$; from then on the
CPE buffer is drained at rate $R_p$.}
\label{fig:hill}
\end{figure}
\subsection{Drop-point Buffer Occupancy of a Single CPE}
\label{bufsingleCPE:sec}
Based on the basic timing analysis in the preceding section,
we characterize the buffer occupancy due to a single CPE $c$
in the drop-point.
The buffer occupancy grows at rate $R_d$ [bit/s]
from arrival instant $\alpha_c$ of the CPE~$c$ upstream transmission to
the drop point until the starting instant $\mu_c$ of the ONU upstream
transmission.
From instant $\mu_c$ on the drop-point buffer drains at
rate $R_p - R_d$ up to instant $\omega_c$,
when the CPE transmission has been completely received at the drop-point.
From instant $\omega_c$ through the end of the ONU upstream transmission
at $\beta_c$, the buffer drains at rate $R_p$.
Since $R_p > R_d$, the maximum buffer occupancy $B_{\max, c}$ occurs
at time instant $\mu_c$ when the ONU starts to serve (transmit)
the CPE traffic.
The drop-point has been receiving CPE data at rate $R_d$ since time instant
$\alpha_c$, resulting in
\begin{eqnarray}  \label{Bmax:eqn}
B_{\max, c} = (\mu_c - \alpha_c) R_d  = G_c - \frac{R_d}{R_p} (G_c - M).
\end{eqnarray}
Thus, the buffer occupancy $B_c(t)$
of the drop point buffer associated with CPE~$c$ is
\begin{eqnarray}
B_c(t)=\begin{cases}  
                      R_d (t-\alpha_c) & t\in[\alpha_c,\mu_c]\\
     B_{\max, c} - (R_p - R_d) \, (t - \mu_c)  & t \in [ \mu_c, \omega_c]\\
     M - R_p \, (t - \omega_c)  & t\in[\omega_c, \beta_c],
    \end{cases}
\end{eqnarray}
and zero otherwise.
For joint CPE buffering in the drop point (ONU),
the superposition of the buffer occupancies
\begin{eqnarray}  \label{Bt:eqn}
B(t) := \sum_c B_c(t)
\end{eqnarray}
characterizes the occupancy level of the shared ONU buffer.
The maximum of $B(t)$ is the maximum ONU buffer occupancy.

\subsubsection{PON Segment Packet Delay}  \label{pondelay:sec}
Considering maximum sized packets,
the first packet of a given CPE upstream transmission is completely received
by the drop point (ONU) at time instant $\alpha_c + M/R_d$.
This first packet has to wait (queue) at the drop point until
its transmission over the PON upstream wavelength channel commences
at time instant $\mu_c$.
Thus, the queueing delay is $\mu_c - \alpha_c -M/R_d$,
which can be expressed in terms of the maximum CPE buffer occupancy
$B_{\max,c}$ (\ref{Bmax:eqn}) as $(B_{\max, c} - M)/ R_d$.
The last packet of the CPE upstream transmission, which is completely
received by the ONU at time instant $\omega_c$ (\ref{omega:eqn}),
does not experience any queueing delay.
Each packet experiences that transmission delay $M/R_p$
and propagation delay $\tau$ of the PON.
Summing these delay components gives the total PON delay for a packet;
and averaging over the packets in the CPE upstream transmission
leads to the average packet delay on the PON segment.

\subsection{ONU:CPE Polling with Segregated CPE Transmissions on PON}
\label{cpesepspec:sec}
In this section, we specify the Gated ONU:CPE polling protocol with
segregated CPE transmissions on the PON upstream channels.
That is, the data of each DSL CPE is transmitted
upstream in its own sub-window of the overall ONU upstream
transmission window.
We consider $E$ CPEs attached to a given ONU.
The ONU sends the $E$ CPE data transmissions successively
according to a prescribed transmission order, as
specified in Section~\ref{cpeorder:sec}, over the
PON upstream wavelength channel.
The CPEs time (schedule) their transmissions
as specified in Section~\ref{cpesegtime:sec} to ensure that
the CPE data arriving at rate $R_d$ to the ONU can be transmitted
without interruptions at the full PON rate $R_p,\ R_p > R_d,$ to the OLT.

\subsubsection{CPE Polling Order} \label{cpeorder:sec}
The detailed analysis of the polling time with two CPEs in Appendix~1
indicates that the transmission order of CPE~1 followed by CPE~2 results in
shorter cycle duration if
\begin{eqnarray}
G_1 < G_2 + 2 \frac{\delta_2 - \delta_1}{\frac{1}{R_d} - \frac{1}{R_p}}.
\end{eqnarray}
That is, transmitting the traffic from the CPE
with the smaller grant size $G_1$ on the upstream PON channel before the
CPE with the larger grant $G_2$ generally
reduces the cycle duration, provided the
round-trip propagation delays $\delta_1$ and $\delta_2$
between the ONU and the two CPEs are not too different.
Typically, the CPEs are all in close vicinity of the ONU, thus
the round-trip propagation delay differences are often
negligible, even when scaled by the $1/(\frac{1}{R_d} - \frac{1}{R_p})$ factor.
For the remainder of this study we consider therefore
the CPE transmission order $c = 1,\ c = 2, \ldots, c = E$
with $G_1 \leq G_2 \cdots \leq G_E$ on the PON upstream transmission channel.

\subsubsection{CPE Transmission Timing} \label{cpesegtime:sec}
We derive the earliest time instant $\mu_{1, 2, \ldots, E}$
that the ONU can start
upstream transmission such that all $E$ CPE data sets arrive in
time to the drop-point for the ONU to continuously transmit at rate
$R_p$.
Specifically, we prove the following theorem:
\begin{theorem}
In order to meet the constraint of continuous (back-to-back) transmission
of the data from CPEs $1, 2, \ldots, E$ in separate sub-transmission windows
at the PON rate $R_p$, the ONU can start transmission
at the earliest at time instant
\begin{eqnarray}
 \mu_{(E)} \!\!\! &=& \!\!\! \mu_{(E-1)}  \nonumber \\
  && \!\!\!\!\!\!\!\!\!\!\! \!
  + \max \left( 0 ,\ \mu_E - \mu_{(E-1)} - \frac{\sum_{c = 1}^{E-1} G_c}{R_p}
    \right),  \label{mu12e:eqn}
\end{eqnarray}
whereby the ONU transmission starting instant $\mu_c,\ c = 1, 2, \ldots, E$,
for an individual CPE $c$ is given by Eqn.~(\ref{muc:eqn}).
\end{theorem}
\begin{proof}
We consider initially two CPEs $c = 1$ and $c = 2$.
Considering each of these two CPEs individually,
Eqn.~(\ref{muc:eqn}) gives the respective time instants $\mu_1$ and
$\mu_2$ when ONU service could at the earliest commence, when considering
a given CPE in isolation.

 There are two cases:
 If $\mu_1 + G_1/R_p > \mu_2$, then the earliest instant for the
 continuous ONU transmission to commence is $\mu_1$. This is because the
 transmission of the data from CPE $c = 1$ takes longer
than CPE $c=2 $ needs to get its data ``ready'' for ONU transmission.

 If, on the other hand, $\mu_1 +G_1/R_p < \mu_2$,
then the ONU transmission of CPE $c = 1$ data must be delayed
in order to avoid a gap between the end of the ONU transmission of the
CPE $c = 1$ data and the start of the ONU transmission of the
CPE $c = 2$ data.
The earliest instant for the continuous ONU transmission to commence is
$\mu_2 - G_1/R_p$, which gives the ONU just enough time to transmit
the CPE $c = 1$ data before the CPE $c = 2$ data is ``ready'' for ONU
transmission.
In summary, the two cases for $E = 2$ CPEs result in the earliest
start time
\begin{eqnarray} \label{induct0:eqn}
 \mu_{(2)} = \max \left( \mu_1,\ \mu_2 - \frac{G_1}{R_p} \right)
\end{eqnarray}
for continuous ONU transmission at rate $R_p$.

We proceed to the general case of $E, \ E > 2$ CPEs by induction:
Consider the continuous (back-to-back) ONU transmission of
CPE $c = 1$ and CPE $c = 2$ data as one CPE transmission with
earliest ONU transmission instant (when considered individually)
$\mu_{(2)}$.
Next, we consider this back-to-back CPE $c = 1$ and $c=2$ data
as well as the CPE $c = 3$ data. Analogous to (\ref{induct0:eqn}),
we obtain the earliest starting instant of the continuous
ONU transmission of the data from CPEs $c = 1, 2$, and 3:
\begin{eqnarray}
 \mu_{(3)} =  \max \left( \mu_{(2)} ,\ \mu_3 - \frac{G_1+G_2}{R_p} \right).
\end{eqnarray}
Proceeding to the induction step with the continuous ONU transmission
of the CPE $c = 1, 2, \ldots, E-1$ data with earliest transmission instant
$\mu_{(E-1)}$ as well as the CPE $c = E$ data results in the
earliest transmission instant given by Eqn.~(\ref{mu12e:eqn})
\end{proof}

The sub-transmission window of CPE $c = 1$ starts at $\mu_{(E)}$,
while CPE $c = 2$ starts when the ONU transmission of CPE $c = 1$ data is
complete. Generally, the starting instants of the segregated CPE
sub-transmission windows $c = 1, 2, \ldots, E$ are
\begin{eqnarray}
  \mu_c^s = \mu_{(E)} + \sum_{i = 1}^{c-1} \frac{G_i}{R_p}.
\end{eqnarray}
From these starting instants $\mu_c^s$ of the segregated CPE sub-transmission
windows, we find the corresponding starting instants $\sigma_c^s$
of the CPE transmissions by re-tracing the analysis in
Section~\ref{basicindcpe:sec}.
Briefly, for the continuous ONU transmission of the CPE $c$ data at rate
$R_p$ is it sufficient for CPE $c$ to commence transmission
$G_c/R_d + M/R_p + \delta_c$ before the end of the ONU transmission at instant
$\mu_c^s + G_c / R_p$, i.e.,
\begin{eqnarray} \label{sigcs:eqn}
  \sigma_c^s = \mu_c^s + \frac{G_c-M}{R_p} -  \frac{G_c}{R_d} - \delta_c.
\end{eqnarray}
Starting the CPE transmissions at $\sigma_c^s$ instead of
the earliest possible $\sigma_c$ (\ref{sigmac:eqn})
for an individual transmission reduces the drop-point buffer occupancy.

\subsection{ONU:CPE Polling with Multiplexed CPE Transmissions on PON}
\label{onucpemux:sec}
In this section, we specify the ONU:CPE polling protocol with statistical
multiplexing of the packets from the individual CPEs
in the ONU upstream transmission window.
All DSL CPEs attached to the same drop-point statistically multiplex their
transmissions into a joint ONU upstream transmission window
(rather than in the separate sub-windows in Section~\ref{cpesepspec:sec}).
The OLT effectively grants transmission
windows to a given ONU to fit in all the traffic (in randomly
statistically multiplexed order) of the DSL CPEs attached to the
drop-point containing the ONU.

\begin{theorem}
When the aggregate upstream transmission bit rate of the $E$ CPEs at
an ONU is less than the PON upstream transmission bit rate, i.e.,
when $E R_d \leq R_p$, then the ONU can commence the continuous
transmission of the multiplexed CPE data at the earliest at
\begin{eqnarray}  \label{mum:eqn}
\mu^m &=& (E+1) g_p+ \tau + g_d + \max_{c}
      \left( 2 \delta_c + \frac{ G_c}{R_d} \right) \nonumber \\
     && \ \ \ \ \ \ \ \ \ \ \ \ \ \ \
   + \frac{ EM - \sum_{c = 1}^E G_c}{R_p}.
\end{eqnarray}
\end{theorem}

\begin{proof}
The individual CPE upstream transmissions $c = 1, 2, \ldots, E$,
can at the earliest be completely received by the drop point by time
instants $\sigma_c + \delta_c + G_c / R_d$, whereby $\sigma_c$ is given
by Eqn.~(\ref{sigmac:eqn}).
The latest such instant of complete reception of the data from a CPE at
the drop point is
\begin{eqnarray}
\omega = (E+1) g_p+ \tau + g_d + \max_{c}
      \left( 2 \delta_c + \frac{ G_c}{R_d} \right).
\end{eqnarray}
If the aggregate transmission bit rate $E R_d$ of the $E$ CPEs does not
exceed the PON upstream transmission bit rate $R_p$,
the ONU can transmit all multiplexed CPE data upstream such that only
one data packet,
from at most each of the $E$ CPEs,
remains to be transmitted after $\omega$.
Thus, the ONU can complete the upstream transmission by
$\omega + E M / R_p$.
Since the ONU has to transmit a total of $\sum_{c = 1}^E G_c$ bits of CPE data,
the corresponding starting time instant of the ONU transmission must be
$\sum_{c = 1}^E G_c / R_p$ before $\omega + E M / R_p$,
resulting in the transmission start instant given by Eqn.~(\ref{mum:eqn}).
\end{proof}

With the ONU transmission starting at instant $\mu^m$,
the ONU transmission is completed at instant
$\mu^m + \sum_{c=1}^E G_c / R_p$.
All CPE data has to arrive to the drop-point at least
$E M/R_p$ before the ONU transmission
completion instant $\mu^m + \sum_c G_c / R_p$.
CPE $c$ data is completely received by the drop point
$G_c / R_d + \delta_c$ after the CPE transmission
starting instant $\sigma_c^m$.
Thus, CPE $c$ can start transmission at the latest at instant
\begin{eqnarray}  \label{sigcm:eqn}
\sigma_c^m = \mu^m + \frac{ \sum_{c=1}^E G_c - EM }{R_p} - \frac{G_c}{R_d} - \delta_c.
\end{eqnarray}

\section{Performance Evaluation}   \label{sec:experiment}
We conducted a wide set of simulations
to answer three questions of practical interest:
\begin{enumerate}
  \item When is flow control required to provide a specific bound on ONU
    buffer occupancy without loss at the ONU?
  \item When does PAUSE frame flow control fail to provide a specific bound
     on ONU buffer occupancy without loss at the ONU?
  \item What is the range of bounds on ONU buffer occupancy without loss
   at the ONU that can be achieved with Gated flow control?
\end{enumerate}
We used a PON/xDSL hybrid access network simulator that we developed
using the CSIM discrete event simulation library. We considered the
XGPON~\cite{XGPON} protocol for the PON segment and the
VDSL2~\cite{VDSL} protocol for the DSL segment as these two
technologies are being actively deployed in real hybrid access
networks. We set the XGPON upstream bit rate to $R_p = 2.488$~Gbps
and the guard time to 30~ns. The XGPON contained $O = 32$ ONUs, each
with $E = 8$ attached VDSL lines (for a total of 256 CPEs). The
upstream bit rate for each VDSL line was set to $R_d = 77$~Mbps to
achieve a realistic worst-case over-subscription rate of 8x. The OLT
to ONU one-way propagation delays $\tau$ were continuously
distributed between 2.5 $\mu$s (i.e., 500~m) and 100 $\mu$s (i.e.,
20~km). The ONU to CPE  propagation delays $\delta$ are considered
negligible. We set the maximum cycle length to $Z = 3$~ms. The CPEs
independently generated data packets according to a quad mode packet
size distribution with 60~\% 64~Byte packets, 4~\% 300~Byte packets,
11~\% 580~Byte packets, and 25~\% 1518~Byte packets.
Each simulation run for a given traffic load considered $10^8$ packets.

\begin{figure*}[t]
\centering
\begin{tabular}{cc}
\includegraphics[scale=0.7]{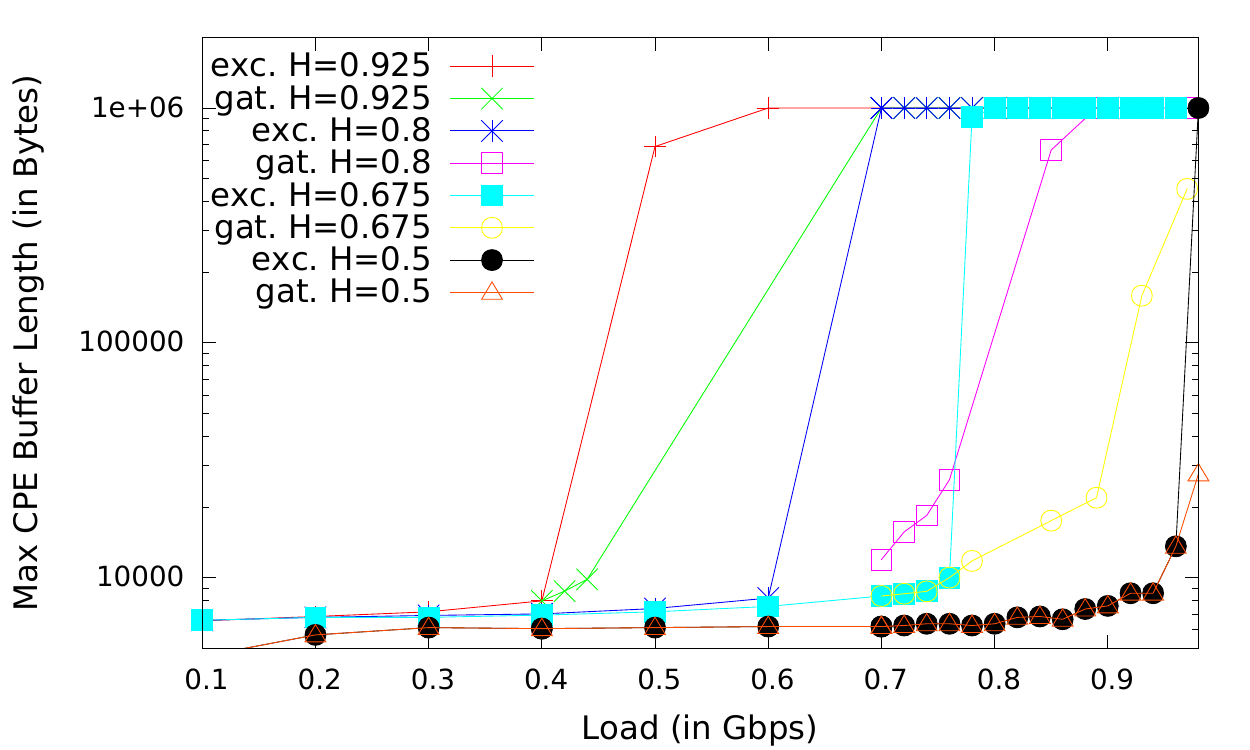} &
\includegraphics[scale=0.7]{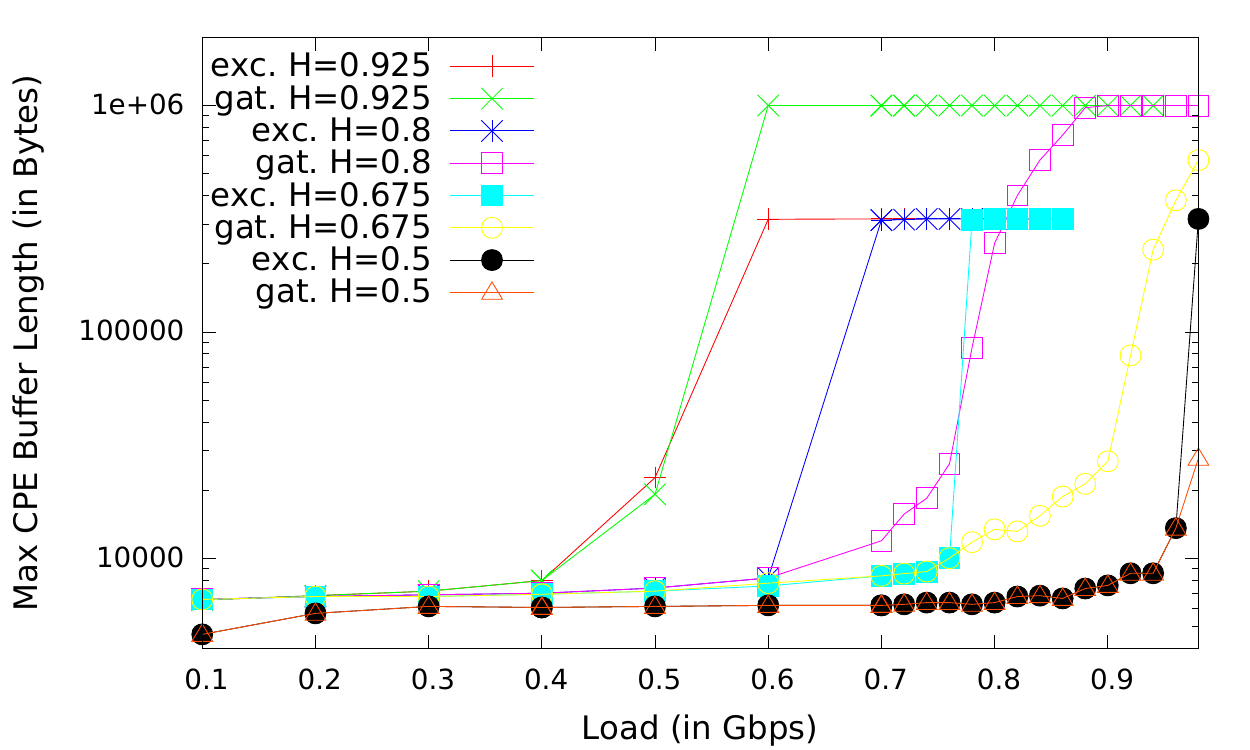} \\
\footnotesize{a) No flow control (max. CPE buff.)} &
\footnotesize{b) PAUSE frame flow control (max. CPE buff.)} \\

\includegraphics[scale=0.7]{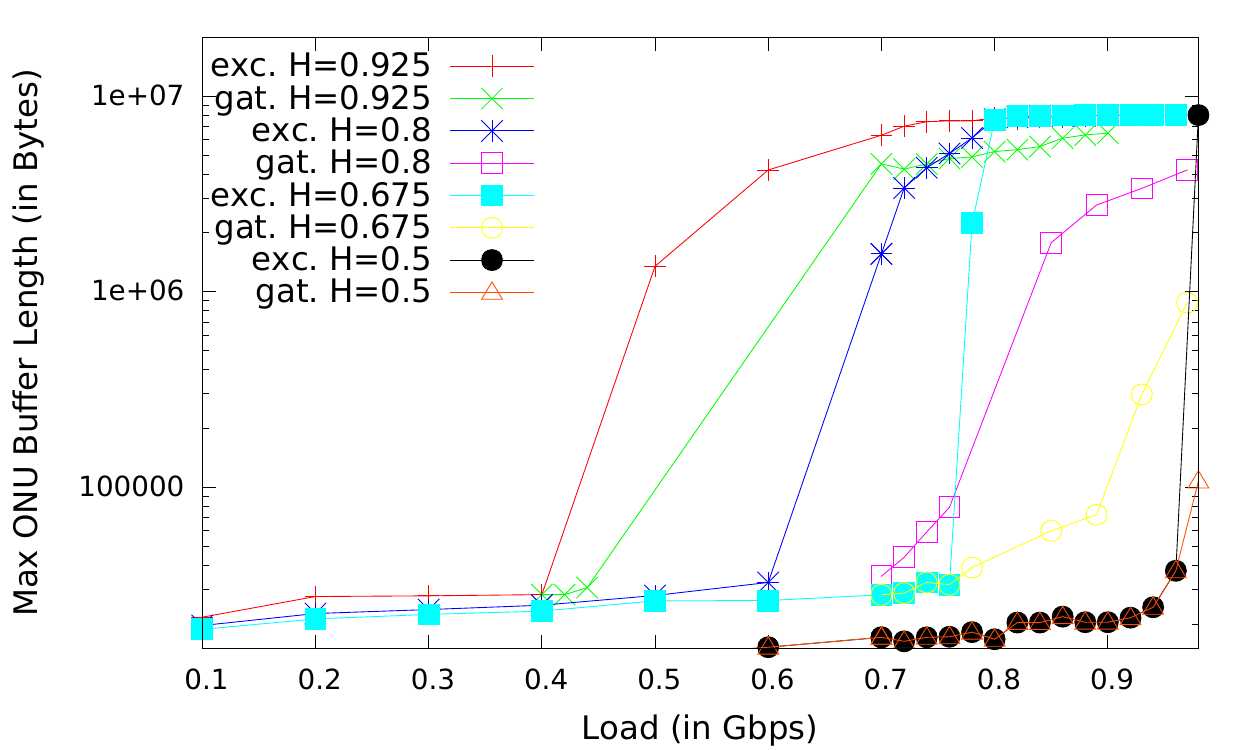} &
\includegraphics[scale=0.7]{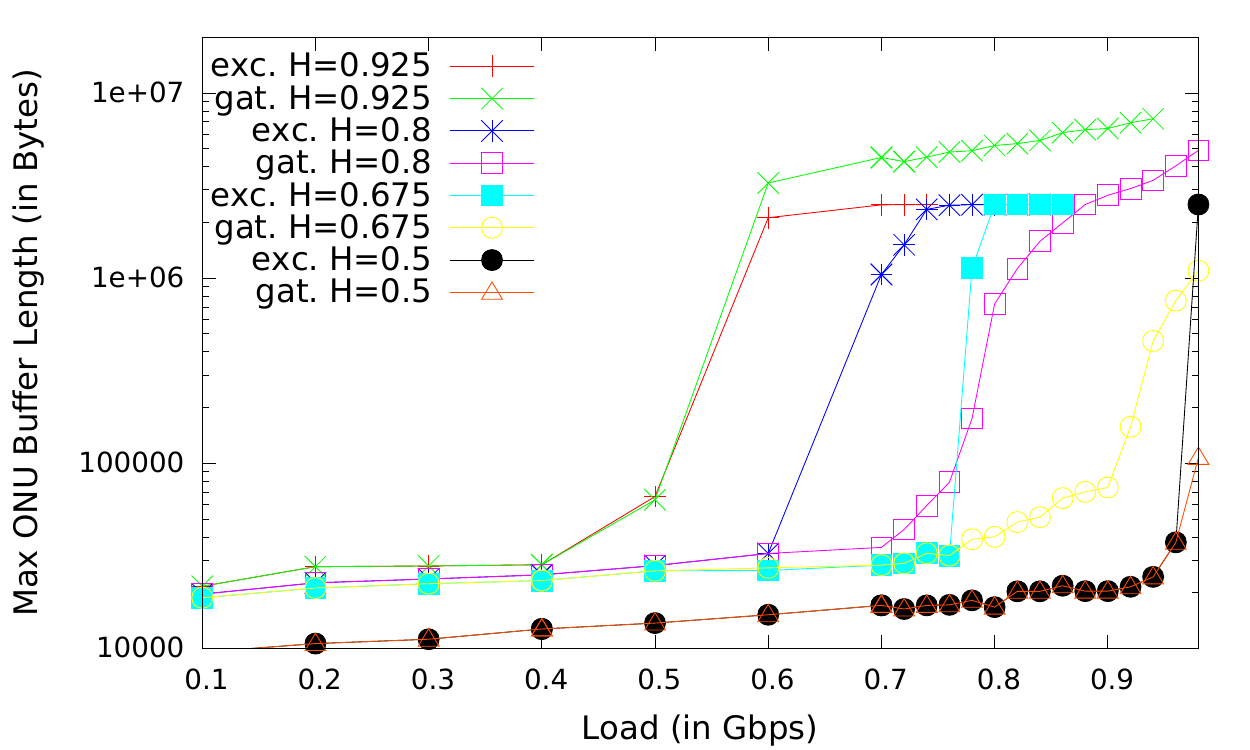} \\
\footnotesize{c) No flow control (max. ONU buff.)} &
\footnotesize{d) PAUSE frame flow control (max. ONU buff.)} \\

\includegraphics[scale=0.7]{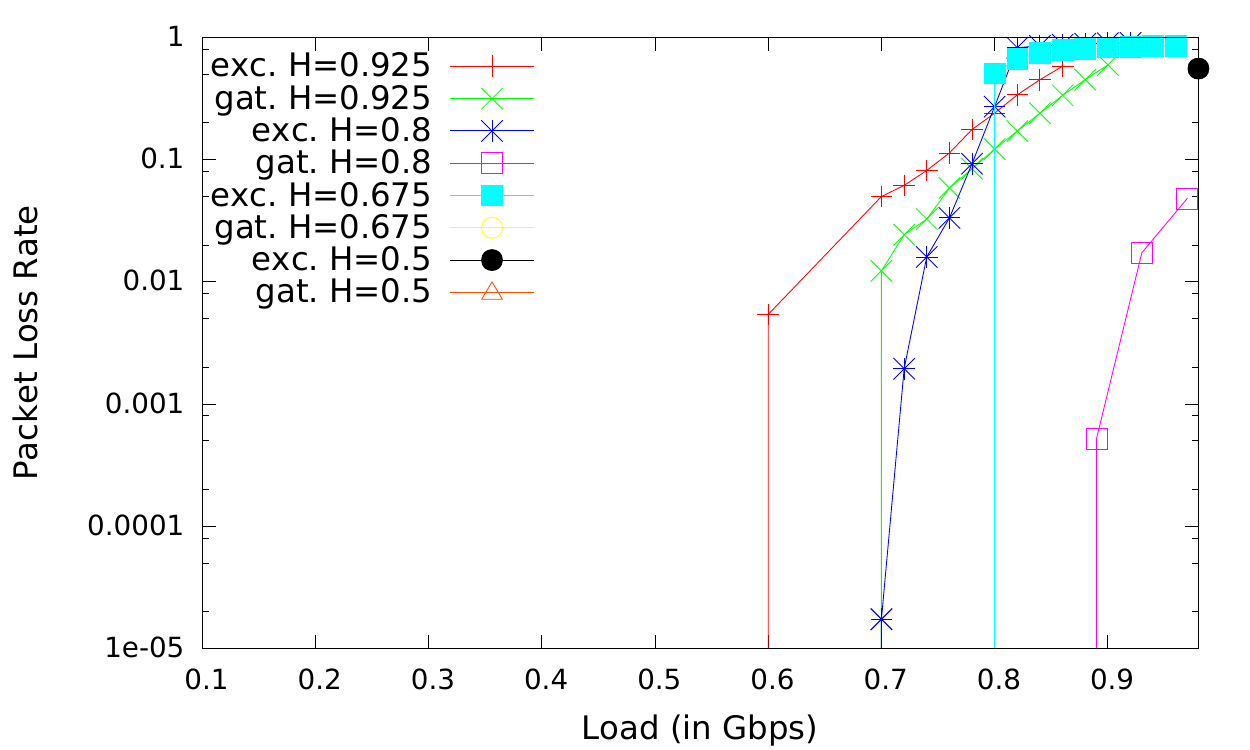} &
\includegraphics[scale=0.7]{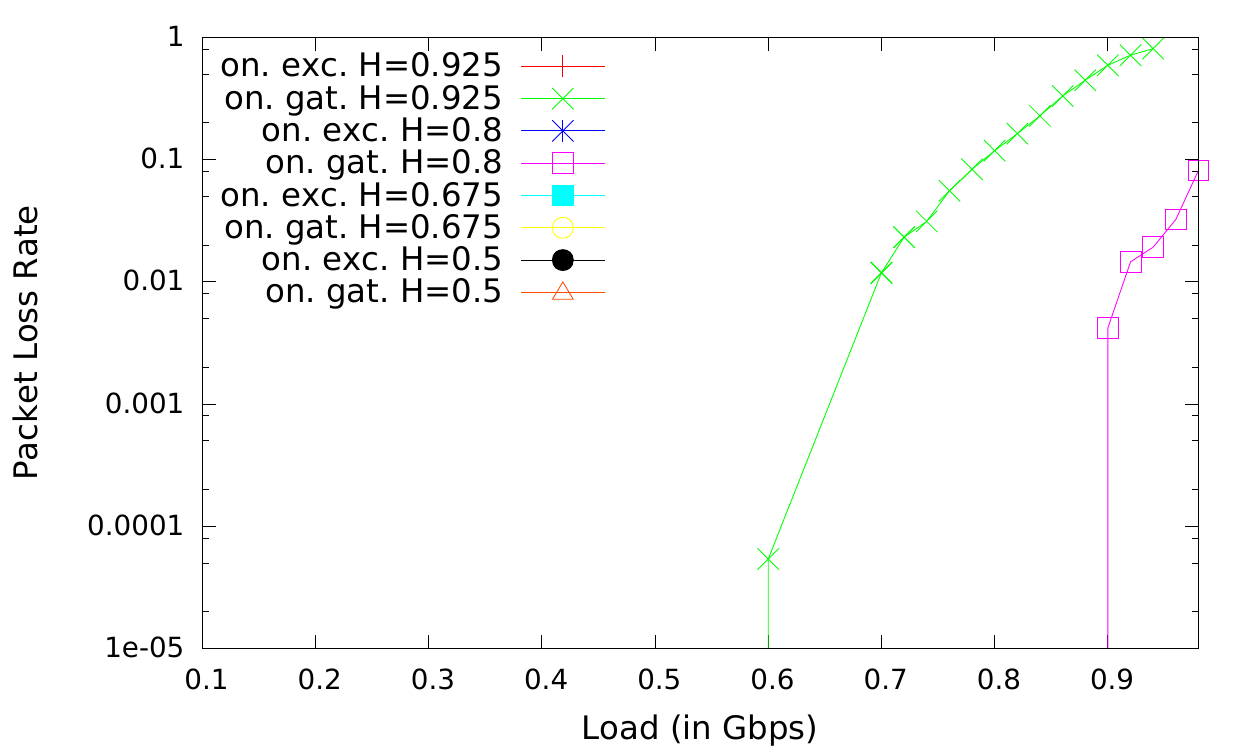} \\
\footnotesize{e) No flow control (pkt. loss rate)} &
\footnotesize{f) PAUSE frame flow control (pkt. loss rate)} \\
\end{tabular}
\caption{Comparison of no flow control vs. ONU polling PAUSE frame flow
control with CPE buffer capacity of 1~MB.}
\label{fig:noflow_vs_pause_100KB}
\end{figure*}
\subsection{No Flow Control}
To answer question 1 we forgo the use of any flow control, utilize
large CPE buffer capacities (1 MBytes), and monitor the
maximum buffer occupancy. The DBA algorithm, source traffic
burstiness, and presented traffic load are factors that will affect
buffer occupancy at the ONU. Therefore, we vary these factors. We
consider the (Online, Gated) and (Online, Excess) DBA algorithms that
have been shown to provide good performance in conventional
PONs~\cite{MR0712}, with a reporting approach akin to~\cite{sku2010dyn}
for the newly generated traffic.
Gated grant sizing assigns each ONU the full bandwidth
request~\cite{KMP0202,ZM0709}.
The employed (Online, Excess) grant sizing
approach assigns each ONU its request up to
the maximum ONU grant size of Limited grant sizing~\cite{KMP0202,ZM0709},
i.e., an $1/O$ share of the total PON upstream transmission capacity $Z R_p$
in a cycle,
plus a $1/O$ share of accumulated unused excess bandwidth (which was
also limited to $Z R_p$)~\cite{MerMcM13}; thus, the total maximum
ONU grant is $2 Z R_p / O$.
We vary the
burstiness of the traffic by using a self-similar traffic source in
which we vary the Hurst parameter from 0.5 (equivalent to a Poisson
traffic source) to 0.925 (equivalent to very bursty traffic).

Fig.~\ref{fig:noflow_vs_pause_100KB}a), c), and e)
contains plots of the maximum buffer occupancies and packet loss rate
versus presented traffic load without the use of flow
control.
The traffic load is represented as a fraction of the full XGPON upstream
transmission rate of 2.488Gbps.
We define the maximum CPE buffer occupancy as the largest (maximum)
of the maximum CPE buffer occupancies $B_{\max, c}$,
see Fig.~\ref{fig:hill} and Eqn.~(\ref{Bmax:eqn}), observed during a very long
simulation considering over $10^8$ packet transmissions.
The maximum ONU buffer occupancy is analogously defined as the largest
aggregate of the CPE buffer occupancies, see Eqn.~(\ref{Bt:eqn}).
Our primary observation from these plots is that the maximum buffer
occupancy increases modestly until a certain ``knee point''
load value and
then increases very sharply. The ``knee point'' load value
depends on both the
DBA algorithm and the burstiness of the source traffic.
If the buffer occupancy below the knee point load value meets
requirements, then flow control can be switched on
just when the knee point load value is reached.
As an example, when using the (Online, Excess) DBA algorithm, the
maximum ONU buffer occupancy value before the knee point
is 32~KB or less and the maximum CPE buffer
occupancy is 10~KB or less.  If 32~KB was the desired upper bound
on the maximum aggregate ONU buffer occupancy, then flow control need
only be activated once the presented load approached 0.94 for
non-bursty traffic ($H=0.5$) or 0.4 for highly bursty traffic
($H=0.925$). Not surprisingly, bursty traffic will require flow control
under wider load conditions than non-bursty traffic.

\subsection{ONU Polling PAUSE Frame Flow Control}
To answer question 2 we use PAUSE frame flow control with a
threshold of 35~\% buffer capacity to trigger the transmission of
PAUSE frames with a duration of 2~ms. A set of experiments, that we
leave out due to space constraints, were conducted to explore that
two-dimensional parameter space of buffer threshold and PAUSE
duration.  Those experiments indicated that (35\%, 2 ms) provided
the best performance.

Figure \ref{fig:noflow_vs_pause_100KB}b), d), and f)
contains plots of the maximum
buffer occupancies and packet loss rates versus presented traffic load
with PAUSE frame flow control.
We observe that the maximum buffer occupancy trends when using PAUSE
frame flow control are similar to when no flow control is used. A
notable exception is that for the (Online, Excess) DBA algorithm,
the maximum CPE buffer occupancy stays below approximately 300~KB
when PAUSE frame flow control is used,
compared to 1~MB (i.e., the full capacity) when no flow control is
used. For the (Online, Gated) DBA algorithm, the maximum CPE
buffer occupancy reaches the 1~MB buffer capacity for highly bursty traffic
($H = 0.8$ and 0.925), regardless of whether PAUSE frame
flow control is used. The unlimited grant sizes of the
(Online, Gated) DBA algorithm appear to undermine the efforts of flow
control.

From the packet loss rate plots in
Figure~\ref{fig:noflow_vs_pause_100KB}e) and f)
we observe that when using the
(Online, Excess) DBA algorithm, PAUSE frame flow control can
eliminate packet losses. On the other hand, for
the (Online, Gated) DBA algorithm with unlimited grant sizes,
PAUSE frame flow control is unable
to lower the packet loss rate for the bursty $H = 0.8$ and 0.925 traffic.

\begin{figure*}[t!]
\centering
\begin{tabular}{cc}
\includegraphics[scale=0.7]{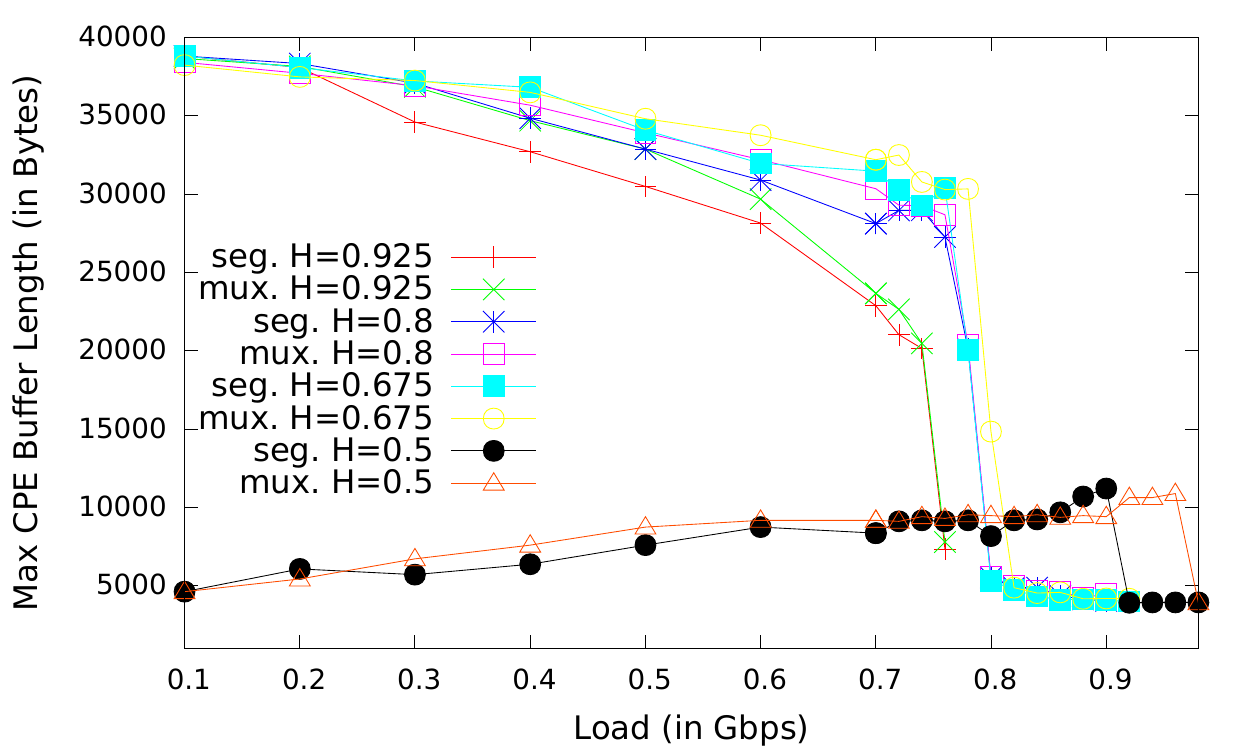} &
\includegraphics[scale=0.7]{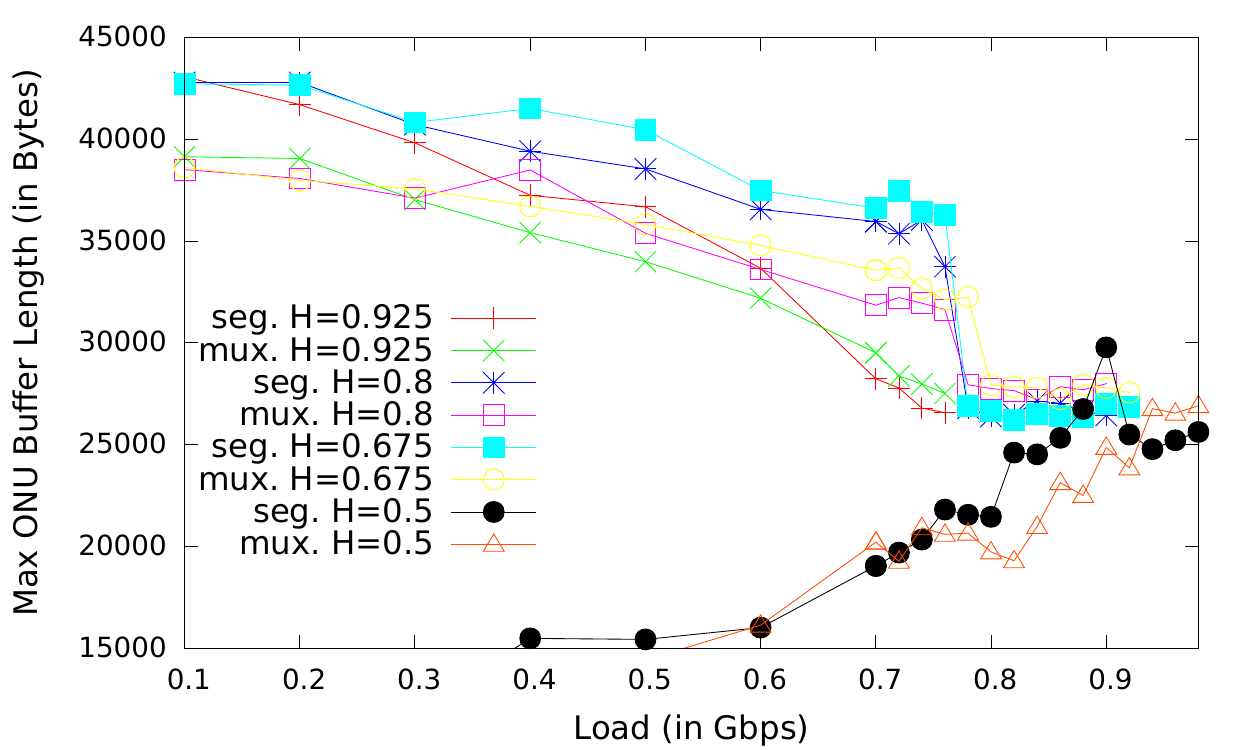} \\
\footnotesize{a) Max. CPE buffer occupancy} &
\footnotesize{b) Max. ONU buffer occupancy}
\end{tabular}
\caption{Maximum occupancies of CPE and ONU buffers for GATED Flow
Control approaches ONU:CPE:seg and ONU:CPE:mux with (Online, Excess)
dynamic bandwidth allocation (DBA) on PON for different levels of
traffic burstiness (i.e., different Hurst parameters $H$).}
\label{fig:gated_excess}
\end{figure*}
\subsection{GATED ONU:CPE Polling Flow Control}
To answer question 3 we present results for the two Gated ONU:CPE
polling flow control protocols introduced in
Section~\ref{onucpe_poll:sec}, namely segregated (ONU:CPE:seg) and
multiplexing (ONU:CPE:mux) polling flow control. We continue to
consider the (Online, Excess) sizing for the ONU grants. A given ONU
grant is distributed to the CPEs according to the equitable
iterative excess distribution method~\cite{AYDA1103,bai2006fair},
which fairly divides the ONU grant among the CPEs, allowing CPEs
with high traffic loads to utilize the unused fair shares of the low
traffic loads.
Figure~\ref{fig:gated_excess} contain plots of the maximum buffer occupancies
and average packet delays as a function of load.

\subsubsection{Maximum CPE and ONU Buffer Occupancies}
We observe from Figure~\ref{fig:gated_excess}
that for low loads of bursty traffic with Hurst parameters $H > 0.5$,
the maximum CPE and ONU buffer occupancies are approximately twice
the maximum ONU grant size of a Limited DBA grant
sizing at low traffic loads, i.e., approximately
$2 Z R_p / O$.
At low bursty traffic loads it is likely that only very few CPEs
(that are attached to only a few ONUs)
generate a traffic burst at a given time, while the other CPEs have no
traffic.
This permits the ONUs with attached CPEs with a traffic burst
through the considered Online
Excess DBA mechanism~\cite{AYDA1103,bai2006fair,MerMcM13}
to utilize the excess bandwidth allocation from the ONUs without
traffic bursts. The considered Online Excess DBA limits
the excess allocation from other ONUs to a given ONU to once the
maximum Limited DBA grant size.
Thus, if a single CPE at an ONU generates a traffic burst,
the CPE is allocated a grant of twice the maximum Limited DBA grant size,
resulting in correspondingly large maximum CPE and ONU buffer occupancies.
(The ONU buffer occupancies slightly above 40~kB are due to small
residual backlog from preceding cycles due to the
different DSL and PON framing mechanisms, see Section~\ref{cpesize:sec}.)

Interestingly, we observe from Figure~\ref{fig:gated_excess} that
the maximum CPE and ONU buffer occupancies in the bursty
($H > 0.5$) traffic scenarios decrease with increasing
traffic load.
As the traffic load increases, more and more CPEs
have backlogged (queued) traffic bursts.
When all ONUs have some CPEs with backlogged traffic, there is no more
excess allocation from ONUs with little or no traffic backlog to
ONUs with large traffic backlog.
Thus, the Online Excess DBA mechanism turns into the Online Limited DBA
mechanism and allocates to each ONU the maximum Limited DBA ONU grant size.
Thus, as the traffic load increases, the traffic amount
transmitted upstream on the PON bandwidth is more equally distributed among
the ONUs as more and more ONUs have CPEs with backlogged traffic bursts.
In turn, the grant allocation to a given ONU is more
equally divided among its attached CPEs as the
traffic load increases and more and more CPEs at an ONU have backlogged
traffic bursts.

For the Poisson traffic scenario ($H = 0.5$), we observe from
Fig.~\ref{fig:gated_excess} that the CPE and ONU buffer occupancies
continuously increase with increasing traffic load (except for a
drop in CPE buffer occupancy at very high loads). In contrast to
bursty traffic sources that generate bursts of several packets at a
time, Poisson traffic sources generate individual data packets.
These individually generated packets are uniformly distributed
(spread) among the CPEs, and correspondingly the ONUs. Thus, there
is essentially no excess allocation among ONUs at low load levels
and the maximum CPE and ONU buffer occupancies grow gradually with
increasing traffic load. (For high load levels there is some excess
allocation, which decreases at very high loads as all CPEs and ONUs
have backlogged traffic, resulting in the CPE buffer occupancy drop
at very high loads.)

In additional simulations, we observed that for the Online, Gated PON DBA,
which grants the ONUs the full bandwidth requests~\cite{ZM0709},
the maximum CPE and ONU buffer occupancies depend mainly on the burstiness
of the traffic: around 10~kBytes for Poisson traffic and
on the order of 10~MBytes  for bursty traffic with $H > 0.5$,
for the considered network scenario.
In contrast, for the Online, Limited PON DBA, which strictly
limits the grant allocation to an ONU to a prescribed limit $Z R_p/O$
(and does not permit re-allocations among ONUs which are
possibly in the Online, Excess PON DBA)~\cite{ZM0709},
we have observed that the maximum CPE and ONU buffer occupancies
are generally bounded by the maximum ONU grant size
$Z R_p/O$~\cite{MGMR1115}.
Thus, our extensive simulations have validated that Gated ONU:CPE polling
flow control effectively limits the maximum CPE and ONU buffer occupancies
through the employed grant sizing mechanisms.

\begin{figure*}[t]
\centering
\begin{tabular}{cc}
\includegraphics[scale=0.7]{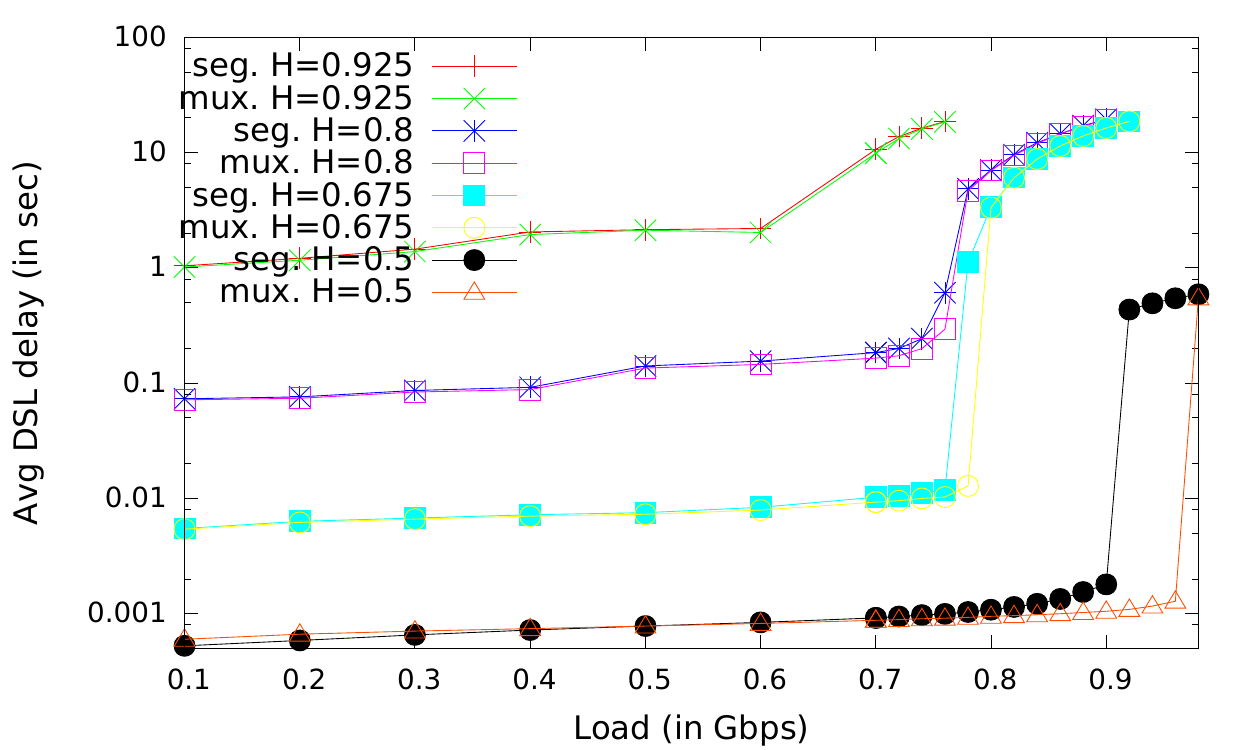} &
\includegraphics[scale=0.7]{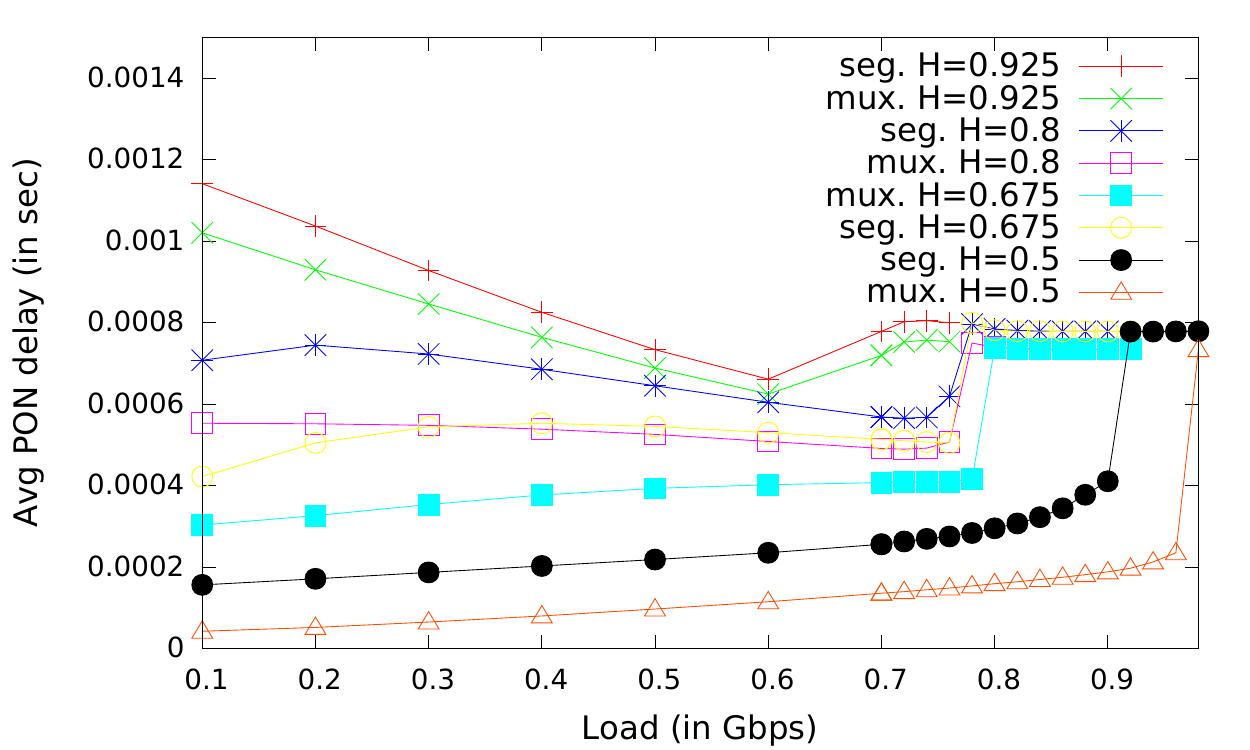} \\
\footnotesize{a) DSL delay} &
\footnotesize{b) PON delay}
\end{tabular}
\caption{Average packet delays on DSL and PON segments for
GATED Flow Control approaches ONU:CPE:seg and ONU:CPE:mux with
(Online, Excess) DBA on PON for different level of
traffic burstiness (i.e., Hurst parameter $H$). }
\label{fig:gated_excess:delay}
\end{figure*}
\subsubsection{DSL and PON Delay}
We observe from Figure~\ref{fig:gated_excess:delay} that the DSL
delay component from the instant of packet generation to the
complete packet reception at the drop point (ONU) increases first
slowly for low loads. Then, for moderately high loads above 0.6, we
begin to observe rapidly increasing DSL delays, first for the highly
bursty $H = 0.925$ traffic and then at higher loads above 0.75 for
the $H = 0.8$ and $H = 0.675$ traffic scenarios. The DSL delays for
these $H > 0.5$ scenarios shoot up to values above 18~s (i.e.,
outside the plotted range) as the traffic bursts overwhelm the
system resources. In contrast, for Poisson traffic, we observe
steadily increasing delays that remain below 1~s even for very high
traffic loads. We also observe that the ``mux'' approach, which
multiplexes upstream transmissions from different CPEs on the
upstream PON wavelength channel achieves lower delays than the
``seg'' approach with segregated CPE upstream transmissions on the
PON. The DSL delay reduction achieved with the multiplexing approach
is particularly pronounced for high Poisson traffic loads, where
the multiplexing approaches reduces the DSL delay by over 0.5~s
compared to the corresponding delay with the segregated approach.

The PON segment delay of a packet from the instant of
packet reception at the drop point (ONU) to the instant
of packet reception at the ONU depends on the
CPE buffer occupancies, as analyzed in Section~\ref{pondelay:sec}.
Essentially, for the segregated CPE transmission approach,
the average PON packet delay corresponds directly
(is proportional) to the
average of the maximum CPE buffer occupancies $B_{\max, c}$
across the individual polling cycles.

We observe from Fig.~\ref{fig:gated_excess:delay} initially (in the low load
region)
decreasing PON delay with increasing load
for the highly bursty $H = 0.925$ traffic,
The other traffic scenarios give initially slowly increasing
PON delays that rapidly increase for high loads in the 0.75--0.95 load range
and then level out.
For the very bursty $H = 0.925$ traffic, the individual
(average)
maximum CPE buffer occupancies $B_{\max,c}$ of payload data packets
(i.e., ignoring the drop point buffer occupancy of CPEs sending only
Report control packets) are initially very large
due to the traffic bursts at individual CPEs and associated ONUs,
which receive excess allocations from the other ONUs
(similar to the dynamics for low loads in Fig.~\ref{fig:gated_excess}).
These excess allocations diminish as CPEs at more and more ONUs get backlogged,
resulting in a decrease of the average maximum CPE buffer occupancies,
and correspondingly a decrease of the average PON packet delay.

For the other traffic scenarios with $H = 0.8$ and lower,
the burstiness is less pronounced, avoiding a decrease of the average
maximum CPE buffer occupancy for increasing loads in the low load region,
whereas the
largest (across a long simulation run) maximum CPE buffer occupancy
does exhibit a significant decrease, see Fig.~\ref{fig:gated_excess}.
For very high loads, the average PON packet delay,
which is proportional to the average maximum CPE buffer occupancy,
levels out around 0.8~ms.
This leveling out is analogous to the leveling out of the largest
maximum CPE buffer occupancy in Fig.~\ref{fig:gated_excess}).
We note that the average PON packet delay of roughly 0.8~ms
is substantially longer than the maximum PON packet delay
obtained with the delay analysis in Section~\ref{pondelay:sec}
for the maximum CPE buffer occupancy of roughly 5~kB for very high loads in
Fig.~\ref{fig:gated_excess}.
The analysis in Section~\ref{pondelay:sec} neglects the small residual
drop point buffering. However, the relatively few packets that
make up the residual buffering have to wait approximately the full
cycle length $Z = 3$~ms for upstream transmission in the next cycle;
thus, substantially increasing the mean PON packet delay.
Nevertheless, due to the flow control back-pressuring the data into the CPEs
until an ONU grant can accommodate the CPE data transmissions, the PON segment
delays are minuscule compared to the DSL segment delays.

We observe from Fig.~\ref{fig:gated_excess:delay} that
multiplexing CPE transmissions gives generally lower PON delays
than segregating CPE transmissions.
The delay analysis in Section~\ref{pondelay:sec} applies directly
to segregated CPE transmissions in that the CPE buffer
in the drop point is filled
at the rate $R_d$ of a single DSL line. The CPE buffer is filled
until the full optical transmission rate $R_p > R_d$ can be
sustained for the transmission of all $E$ CPE data sets over the PON.
When multiplexing CPE transmissions, multiple DSL lines supply
data to the drop point. Thus, the PON transmission can commence earlier,
resulting in shorter queueing delays for the first packets that arrived from
the CPEs to the drop point.

\section{Conclusion}  \label{sec:conclusion}
We have examined the buffering in the drop-point device connecting
the relatively low-transmission rate xDSL segment to the relatively
high-transmission bit rate PON segment in a hybrid PON/xDSL access
network. We found that the drop-point device experiences very high
buffer occupancies on the order of Mega Bytes or larger when no flow
control or when flow control with the standard Ethernet PAUSE frame
are employed. In an effort to reduce the buffer occupancies in the
drop points and thus to reduce the energy consumption of the drop
point devices, which are typically reverse powered from subscribers,
we introduced Gated ONU:CPE polling flow control protocols. We
specified the timing (scheduling) of these Gated ONU:CPE polling
flow control protocols for two types of upstream transmission:
segregated CPE sub-windows or
 multiplexed CPE transmissions within an ONU
upstream transmission window. Through extensive simulations for a
wide range of levels of traffic burstiness, we verified that the
Gated ONU:CPE polling protocols effectively limit the drop-point
buffering in individual CPE buffers or an aggregated ONU buffer. The
maximum CPE and ONU buffer occupancies correspond approximately to
the grant size limits of the polling-based medium access executed at
the OLT. Through adjusting the ONU and CPE grant sizes in the
proposed Gated ONU:CPE polling flow control protocols, the OLT can
effectively control the buffering in the drop-point devices.

One important direction for future research on
hybrid access networks is to extend the hybrid access
network evaluation to the local wired and wireless networks that
interconnect with the access network at the CPE.
Excessive buffering in the CPEs could be mitigated by further back-pressuring
the data transmissions to the gateways or host whose applications
generate large traffic flows.
Another important direction for future research is to examine
control mechanisms through software defined networking in hybrid
access networks~\cite{ker2014sof}.

\section*{Acknowledgment}
This material is based upon work supported by the National Science
Foundation under Grant No. CNS-1059430 and by a gift from Huawei
Technologies.

\vspace{\baselineskip}

\begin{figure*}[t]
\centering
\begin{tabular}{ccc}
\includegraphics[scale=0.65]{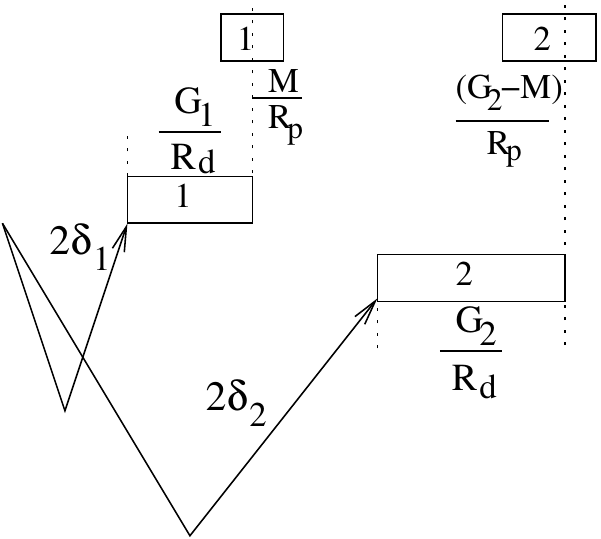}
&  \includegraphics[scale=0.65]{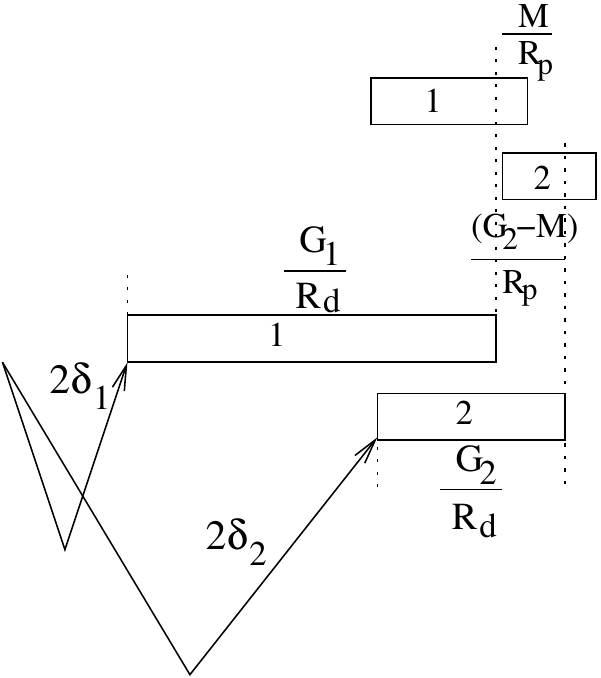} &
\includegraphics[scale=0.65]{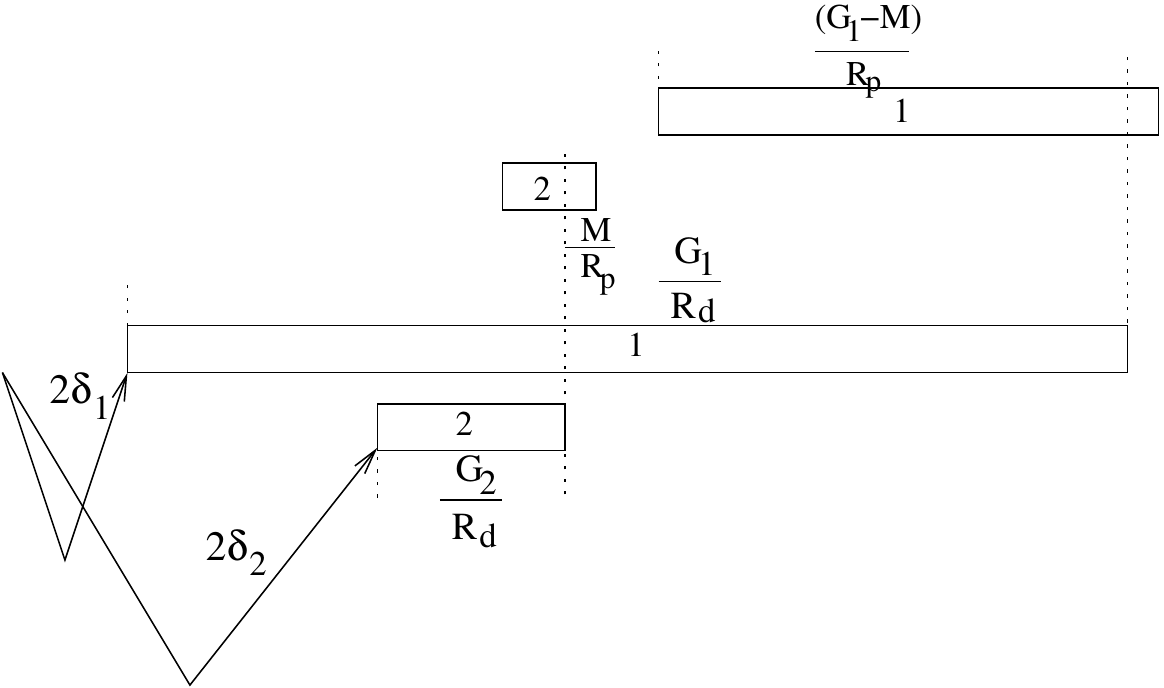} \\
\scriptsize{(a) $G_1 < G_1^{\mathrm{th1}}$}  &
   \scriptsize{(b) $G_1^{\mathrm{th1}} < G_1 < G_1^{\mathrm{th2}}$}  &
      \scriptsize{(c) $G_1 > G_1^{\mathrm{th2}}$}
\end{tabular}
\caption{Illustration of cases for analysis of minimum completion time
for two CPEs (CPE 1 and CPE 2) with segregate sub-windows in the PON grant.
This illustration shows the round-trip propagation delays
$2\delta_1,\ 2\delta_2$ on the DSL networks as well as the
DSL upstream transmission delays $G_1/R_d$ and $G_2/R_d$.
The PON upstream transmission delays $(G_1 - M)/R_p$
and  $(G_2 - M)/R_p$ can be masked by the DSL upstream transmissions
and influence when the PON upstream transmissions can commence.
The PON upstream transmission delays $M/R_p$ occur
after the DSL upstream transmission is complete.
}
\label{Tc_3cases:fig}
\end{figure*}
\section*{Appendix: Analysis of CPE Transmission Ordering for Two CPEs}
We assume for the following analysis without loss of
generality that CPE 1 has a smaller propagation delay to
the drop-point device than CPE 2 (i.e., $\delta_1 \leq \delta_2$).
We analyze the minimum delay $T$ for complete reception of both
upstream transmissions at the OLT.
There are three main cases for evaluating $T$,
as illustrated in Fig.~\ref{Tc_3cases:fig}:
\paragraph{Case small $G_1$, see Fig.~\ref{Tc_3cases:fig}(a)}
There is a gap between the CPE partitions on the PON since
$G_1$ is too small to mask the time until $G_2$ is ready for PON
upstream transmission.
(In this small $G_1$ case, the transmission
of CPE~1 could be delayed so as to avoid the occurrence of a gap,
and reduce the time that the ONU buffer holds the CPE~1 data.)
\paragraph{Case medium $G_1$, see Fig.~\ref{Tc_3cases:fig}(b)}
The partitions of CPE~1 and CPE~2 are transmitted back-to-back on the PON.
\paragraph{Case large $G_1$, see Fig.~\ref{Tc_3cases:fig}(c)}
$G_1$ is so large that the PON upstream transmission of $G_2$
is completed before $G_1$ is ready for PON upstream transmission.

We proceed to analyze the transmission order of the CPE transmission
windows on the PON and identify the minimum times for complete
reception of both CPE data transmissions at the OLT.
We denote with 12 the transmission order CPE~1 followed by
CPE~2, and denote 21 for the reverse transmission order.
To reduce clutter in this scheduling analysis, we re-define the
time periods $\beta$ and $\mu$ from Section~\ref{basicindcpe:sec} with
reference to the end of the downstream gate transmission by the ONU.

In order to identify the threshold $G_1^{\mathrm{th1}}$
that distinguishes the small and
medium $G_1$ cases we
initially consider the transmissions of CPE~1 and CPE~2 as completely
independent, i.e., we initially only consider one of these CPE
transmissions at a time.
From Fig.~\ref{Tc_3cases:fig}(a), we note that
the time period from the ending instant of the gate message transmissions
by the ONU to the time instant that
the ONU transmission of CPE~1 data is completed as
\begin{eqnarray}  \label{beta1:eqn}
\beta_1 = 2 \delta_1 + \frac{G_1}{R_d}
           + \frac{M}{R_p}.
\end{eqnarray}
Similarly, we express the time period until the time instant of
the beginning of the CPE~2 data transmission on the PON as
\begin{eqnarray}  \label{alpha2:eqn}
\mu_2 = 2 \delta_2 + \frac{G_2}{R_d} - \frac{G_2- M}{R_p}.
\end{eqnarray}
The transmission of CPE~1 data by the ONU is completed before
the ONU transmission of CPE~2 data can commence if
$\mu_2 > \beta_1$, i.e., if
\begin{eqnarray} \label{Gth1:eqn}
G_1 < G_2 \left( 1 - \frac{R_d}{R_p} \right)
      + 2 R_d (\delta_2 - \delta_1 ) =: G_1^{\mathrm{ th1}}.
\end{eqnarray}
 Thus, for $G_1 < G_1^{\mathrm th1}$, the transmission of CPE~1 data before
CPE~2 data does \textit{not} delay the commencement of CPE~2 data
transmission. Hence, the transmission order 12 achieves the
minimum cycle (completion) time
\begin{eqnarray}
T = 3 g_p + g_d + 2 \tau + 2 \delta_2 + \frac{G_2}{R_d}
           + \frac{M}{R_p}.
\end{eqnarray}

Next, we identify the threshold $G_1^{\mathrm{th2}}$ that distinguishes
the medium and large $G_1$ cases.
We note from Fig.~\ref{Tc_3cases:fig}(c) that
the ONU transmission of CPE~2 data is completed by
\begin{eqnarray}
\beta_2 = 2 \delta_2 + \frac{G_2}{R_d} + \frac{M}{R_p}.
\end{eqnarray}
The ONU transmission of CPE~1 data can commence at the
earliest at time
\begin{eqnarray}
\mu_1 = 2 \delta_1 + \frac{G_1}{R_d} - \frac{G_1 - M}{R_p}.
\end{eqnarray}
For $\mu_1 > \beta_2$, or equivalently, for
\begin{eqnarray}
G_1 > G_2 + 2 \frac{\delta_2 - \delta_1}{\frac{1}{R_d}
          - \frac{1}{R_p}} =: G_1^{\mathrm{th2} }.
\end{eqnarray}
 the ONU transmission of CPE~1 data is completed
before the ONU transmission of CPE~2 data can commence
That is, the CPE~2 data transmission does \textit{not} delay
the CPE~1 data transmission.
Thus, the 21 transmission order gives the minimum
completion time
\begin{eqnarray}
T = 3 g_p + g_d + 2 \tau
      + 2 \delta_1 + \frac{G_1}{R_d} + \frac{M}{R_p}.
\end{eqnarray}
Note also that $G_1^{\mathrm{th1}} \leq G_1^{\mathrm{th2}} \
          \forall \delta_2 \geq \delta_1,\ R_p > R_d$.

We now turn to the medium $G_1$ range illustrated in
Fig.~\ref{Tc_3cases:fig}(b).
We note from the illustration in Fig.~\ref{Tc_3cases:fig}(b) that the completion
time for the 12 transmission order is
\begin{eqnarray}   \label{TcsmallG1:eqn}
T^{12} = 3 g_p + g_d + 2 \tau + 2 \delta_1 + \frac{G_1}{R_d}
           + \frac{M + G_2}{R_p}.
\end{eqnarray}
We similarly obtain the completion time $T_c^{21}$ for the
21 transmission order and note that
\begin{eqnarray}
T_c^{12} &\leq& T_c^{21} \\
\Leftrightarrow 2 \delta_1 + \frac{G_1}{R_d}  + \frac{ G_2}{R_p}
    &\leq & 2 \delta_2 + \frac{G_2}{R_d}  + \frac{ G_1}{R_p}
 \label{detcrit:eqn}\\
\Leftrightarrow  G_1 &\leq& G_1^{\mathrm{th2} }.
\end{eqnarray}
Thus, the transmission order 12 gives the minimum
$T$ if $G_1 \leq G_1^{\mathrm{th2}}$.

In summary, the minimum time period $T$ from the instant
of commencing the transmission of the gate messages from the OLT to
the complete reception of both CPE data transmissions at the OLT is
obtained by the transmission order CPE~1 data followed by CPE~2 data
on the PON for $G_1 \leq G_1^{\mathrm{th2}}$.
For $G_1 \geq G_1^{\mathrm{th2}}$, the reverse transmission order of
CPE~2 data followed by CPE~1 data on the PON minimizes $T$.

\bibliographystyle{IEEEtran}

\end{document}